\documentclass[journal,twoside,web]{ieeecolor}
\usepackage{generic}
\usepackage{cite}
\usepackage{amsmath,amssymb,amsfonts}
\usepackage{graphicx}
\usepackage{textcomp}

\usepackage{subfigure}
\usepackage{float}
\usepackage[ruled]{algorithm2e}
\usepackage{bm}
\usepackage{booktabs} 
\usepackage{color}

\newtheorem{theorem}{Theorem}[section]
\newtheorem{lemma}{Lemma}[section]
\newtheorem{definition}{Definition}[section]
\newtheorem{remark}{Remark}[section]
\newtheorem{proposition}{Proposition}[section]
\newtheorem{assumption}{Assumption}[section]
\newtheorem{corollary}{Corollary}[section]

\def\BibTeX{{\rm B\kern-.05em{\sc i\kern-.025em b}\kern-.08em
    T\kern-.1667em\lower.7ex\hbox{E}\kern-.125emX}}
\begin{document}
\title{Resilient Consensus with Multi-hop Communication}
\author{Liwei Yuan and Hideaki Ishii, \IEEEmembership{Fellow, IEEE}
\thanks{This work was supported in the part by JSPS under Grant-in-Aid for 
	Scientific Research Grant No.~18H01460. The support provided by the China
	Scholarship Council is also acknowledged. }
\thanks{L. Yuan and H. Ishii are with the Department of Computer Science, Tokyo Institute of Technology, Yokohama, 226-8502, Japan. (e-mail: yuan@sc.dis.titech.ac.jp; ishii@c.titech.ac.jp). }
}

\maketitle

\begin{abstract}
In this paper, we study the problem of resilient consensus
for a multi-agent network where some of the nodes might be
adversarial, attempting to prevent consensus by transmitting
faulty values. Our approach is based on that of the so-called
weighted mean subsequence reduced (W-MSR) algorithm with
a special emphasis on its use in agents capable to communicate
with multi-hop neighbors.
The MSR algorithm is a powerful tool for achieving resilient
consensus under minimal requirements for network structures,
characterized by the class of robust graphs. 
Our analysis highlights that through multi-hop communication,
the network connectivity can be reduced especially in comparison
with the common one-hop communication case.
Moreover, we analyze the multi-hop W-MSR algorithm with
delays in communication since the values from different multi-hop
neighbors may arrive at the agents at different time steps. 
\end{abstract}


\section{Introduction}

\IEEEPARstart{C}{yber}
security of multi-agent systems and distributed algorithms has become an important research area in systems control in the last decade. For multi-agent systems, consensus is one of the fundamental problems \cite{bullo2009distributed}, \cite{Lynch}. 
Based on consensus algorithms, various applications and distributed algorithms have been developed to solve different industrial problems, e.g., clock synchronization \cite{kikuya2017fault}, energy management \cite{yang2013consensus}, distributed state estimation \cite{mitra2021}, distributed optimization \cite{nedic2009distributed, sundaram2018distributed}, and so on.
As concerns for cyber security have rised in general, consensus problems in the presence of adversarial agents creating failures and attacks have attracted much attention; see, e.g.,
\cite{dibaji2018resilient, leblanc2013resilient, su2017reaching, nugraha2021dynamic}. 
One class of interdisciplinary problems that has been studied in both control and computer science is that of resilient consensus \cite{dibaji2017resilient, leblanc2013resilient, vaidya2012iterative}. 
In these works, the adversarial agents are categorized into basically two types: Malicious agents and Byzantine agents. These agents are capable to manipulate their data 
arbitrarily and may even collaborate with each other. However, certain differences lie between them since malicious agents are limited as they must broadcast the same messages 
to all of their neighbors, while Byzantine agents can send individual messages to different neighbors (e.g., \cite{leblanc2013resilient, Lynch, teixeira2012attack}). 

Resilient consensus problems under the Byzantine model have a rich history in the area of distributed computing \cite{Lynch}. In \cite{dolev1982byzantine}, it has
been shown that given a network with $n$ nodes and at most $f$ Byzantine nodes, a necessary and sufficient condition to reach resilient exact consensus is that $n\geq3f+1$ and the graph connectivity is no less than $2f+1$. Furthermore, in  \cite{fischer1985impossibility}, it has been reported that under deterministic asynchronous updates, even one misbehaving agent can make it impossible for the system to reach exact consensus. Then, to avoid this constraint of exact consensus, the authors of \cite{dolev1986reaching} introduced the approximate Byzantine consensus problem in complete networks (i.e., all-to-all communication), where the non-adversarial, normal nodes are required to achieve approximate agreement by converging to a relatively small interval in finite time. In \cite{vaidya2012iterative}, the authors proposed a necessary and sufficient condition for the approximate Byzantine consensus problem in networks with general topologies.

In this paper, we study resilient asymptotic (approximate) consensus under malicious model from the viewpoint of the so-called mean subsequence reduced (MSR) algorithms, which are known for the simplicity and scalability (e.g., \cite{azadmanesh2002asynchronous, leblanc2013resilient, vaidya2012iterative, abbas2020interplay, senejohnny2019resilience, wang2020event, wang2021resilient}). 
This line of work has gained much attention in the last decade as the malicious model can be widely assumed for broadcasting networks \cite{leblanc2013resilient}, wireless sensor networks \cite{kikuya2017fault}, and so on. 
A basic assumption in MSR algorithms is the knowledge regarding an upper bound on the maximum number of malicious agents among the neighbors; this bound is denoted by $f$ throughout this paper. Such a bound represents the level of caution assumed by the system operator and can be set based on past experience, with possibly some safety margin.
Then, at each iteration, each node eliminates extreme values received from neighbors to avoid being influenced by such potentially faulty values. In particular, it removes the $f$ largest values and the $f$ smallest values from neighbors.
Moreover, the graph property called robustness is shown to be critical for the network structure, guaranteeing the success of resilient consensus algorithms in static networks \cite{dibaji2017resilient, leblanc2013resilient} as well as time-varying networks \cite{saldana2017resilient}. 
A recent work \cite{usevitch2020determining} attempts to check robustness of given graphs using mixed integer linear programming.
Nevertheless, such robustness requires the networks to be relatively dense and complex. Therefore, how to enhance resilience of more sparse networks without changing the original network topologies has become an urgent problem.

There are several directions for developing solutions to tackle this problem. In \cite{abbas2017improving}, the authors improved the graph robustness by introducing trusted nodes, which are assumed to be fault-free with extra security measures. They also provided another alternative method \cite{abbas2019diversity} to enhance the graph robustness by introducing different types of nodes in the network and by assuming that the attackers can only compromise a certain type of nodes.
On the other hand, the works \cite{yuan2021secure, zhao2018resilient} pursue approaches based on detection of malicious agents in the network. Compared to MSR algorithms, which do not have such detection capabilities, the algorithms are applicable to more sparse networks with the same tolerance of adversaries.
Furthermore, in \cite{su2017reaching}, by introducing multi-hop communication in MSR algorithms, the authors solved the approximate Byzantine consensus problem with a weaker condition on network structures compared to that derived under the one-hop communication model \cite{vaidya2012iterative}. 
In \cite{sakavalas2020asynchronous}, the authors tackled the same problem under asynchronous updates based on rounds, which is different from the asynchrony setting used in this paper (see the discussions in Section VI).
The work \cite{khan2020exact} studied Byzantine binary consensus under the local broadcast model (malicious model) using a flooding algorithm, where nodes relay their values over the entire network.

Multi-hop communication techniques are commonly used in the areas of wireless communication \cite{goldsmith2005wireless} and computer science \cite{Lynch}. 
Such techniques are also used for consensus problems in many works. In \cite{jin2006multi}, a multi-hop relay technique is introduced in the consensus problem to increase the speed of consensus forming. In \cite{zhao2016global}, a similar method based on multi-hop relay is developed to solve the global leader-following consensus problem. Moreover, application of multi-hop communication in wireless sensor networks from the viewpoint of control is investigated in \cite{manfredi2013design}. It is clear that with multi-hop communication, each node can have more information for updating compared to the one-hop case. Thus, the network may have more resilience against adversary nodes. Yet, only few works have looked into resilient consensus with multi-hop communication. 
As mentioned earlier, the work \cite{su2017reaching} is restricted to the Byzantine adversary case. 
To the best of our knowledge, resilient consensus with multi-hop communication under malicious attacks still remains as an open problem.
Moreover, for many applications of multi-agent systems, time delays in the communication among agents can happen naturally \cite{ren2011distributed}, \cite{mesbahi2010graph}. 
Especially, in the multi-hop communication setting, the length of time delays should increase as the messages are relayed by more agents. 
It is hence of significant importance to extend our algorithm to the case of asynchronous updates with time delays. We must note that such a case is not considered in \cite{su2017reaching} and \cite{khan2020exact}.

The main contribution of this paper can be outlined as follows. 
Inspired by the definition of graph robustness of \cite{leblanc2013resilient}, we extend this notion to the multi-hop setting and name it as \textit{robustness with $l$ hops}.
Specifically, we formally characterize the ability of normal nodes to be influenced by normal multi-hop neighbors outside a given node set in the malicious environment.
Furthermore, we provide analysis for the properties of the new notion of robustness.

Unlike in the case with one-hop communication, the MSR algorithm in the multi-hop case may not exclude all the possible effects from malicious nodes if each normal node just eliminates the $f$ largest and the $f$ smallest received values. Since a malicious node can manipulate not only its own value but also the values it relays, such nodes can produce more than $f$ false values even if there are at most $f$ malicious nodes. To completely exclude the effects from malicious nodes, we propose the multi-hop weighted mean subsequence reduced (MW-MSR) algorithm. Normal nodes using the MW-MSR algorithm will exclude the extreme values which are produced precisely by $f$ multi-hop neighbors. To realize this trimming capability in the multi-hop setting requires the notion of \textit{message cover}, which represents the set of nodes intersecting with a given set of different message paths.

In this paper, we consider the malicious model, which is suitable for broadcast network and is different from the Byzantine model studied in \cite{su2017reaching}. 
Then we derive necessary and sufficient graph conditions based on the new notion of robustness with $l$ hops for the proposed MW-MSR algorithms to achieve resilient consensus under synchronous updates and asynchronous updates with time delays in the communication.
Moreover, we present examples to illustrate how multi-hop communication helps to improve graph robustness without changing the network topology. 
As a side result, we prove that for the case of unbounded path length in message relaying, our graph condition is equivalent to the necessary and sufficient graph condition for binary consensus under malicious attacks studied in \cite{khan2020exact}.

The rest of this paper is organized as follows. 
In Section~II, we outline preliminaries on graphs and the system model. In Sections~III and IV, we present the MW-MSR algorithm and define graph robustness with multi-hop communication, respectively.
Then in Sections~V and VI, we derive tight graph conditions under which the MW-MSR algorithms guarantee resilient asymptotic consensus under synchronous and asynchronous updates, respectively.
In Section~VII, we provide some properties of the new robustness and in Section~VIII, we present examples to demonstrate that multi-hop communication can improve robustness of graphs in general.
Lastly, in Section~IX, we conclude the paper.
A preliminary version of this paper appeared as \cite{yuan2021resilient}.
The current paper contains all the proofs of the theoretical results, further discussions, as well as more extensive simulations.

\section{Preliminaries and Problem Formulation}
In this section, we provide preliminaries on the network models
and introduce the basic settings for the resilient
consensus problems studied in this paper.

\subsection{Network Model}
First, we introduce the graph notions used in this paper.
Consider the directed graph $\mathcal{G} = (\mathcal{V},\mathcal{E})$ consisting of the node set $\mathcal{V}=\{1,...,n\}$ and the edge set $\mathcal{E}\subset \mathcal{V} \times \mathcal{V}$. Here, the edge $(j,i)\in \mathcal{E}$ indicates that node $i$ can get information from node $j$. 
A path from node $i_1$ to $i_m$ is a sequence of distinct nodes $(i_1, i_2, \dots, i_m)$, where $(i_j, i_{j+1})\in \mathcal{E} $ for $j=1, \dots, m-1$. Such a path is referred to as an $(m-1)$-hop path (or a path of length $m-1$) and also as $(i_1,i_m)$-path when the number of hops is not relevant but the source and destination nodes are. We also say that node $i_m$ is reachable from node $i_1$. 
An $\mathcal{X}u$-path is a path from a node in set $\mathcal{X}$ to node $u\notin \mathcal{X}$. We also denote the set minus symbol by $\mathcal{X}\setminus\mathcal{Y}$.

For node $i$, let $\mathcal{N}_i^{l-}$ be the set of nodes that can reach node $i$ via at most $l$-hop paths, where $l$ is a positive integer. Also, let $\mathcal{N}_i^{l+}$ be the set of nodes that are reachable from node $i$ via at most $l$-hop paths. 
The $l$-th power of the graph $\mathcal{G}$, denoted by $\mathcal{G}^l$, is a multigraph\footnote[1]{
	In a multigraph, two nodes can have multiple edges between them.} with the same vertices as $\mathcal{G}$ and a directed edge from node $j$ to node $i$ is defined by a path of length at most $l$ from $j$ to $i$ in $\mathcal{G}$. 
The adjacency matrix $A = [a_{ij} ]$ of $\mathcal{G}^l$ is given by $\alpha \leq a_{ij}<1$ if $j\in \mathcal{N}_i^{l-}$ and otherwise $a_{ij} = 0$, where $\alpha > 0$ is a fixed lower bound. We assume that $\sum_{j=1,j\neq i}^{n} a_{ij}\leq 1$ for all $i$. Let $L = [b_{ij} ]$ be the Laplacian matrix of $\mathcal{G}^l$, whose entries are defined as $b_{ii} =\sum_{j=1,j\neq i}^{n}a_{ij}$ and $b_{ij} = -a_{ij}, i\neq j$; we can see that the sum of the elements of each row of $L$ is zero.

\subsection{Multi-hop Communication for Multi-agent Consensus}\label{problemsetting}

Here, we introduce the multi-agent system with multi-hop communication and the update rule used by the agents under no attacks.
Consider a time-invariant network modeled by the directed graph $\mathcal{G} = (\mathcal{V},\mathcal{E})$. Each node $i$ has a real-valued state $x_i[k]$.
The goal of the agents is to arrive at consensus in their state values asymptotically, that is, $|x_i[k] - x_j[k]| \rightarrow 0$ as $k\rightarrow \infty$ for all $i,j\in \mathcal{V}$. This is to be achieved by updating the states at each time step $k$ based on the information exchanged among the nodes. Their initial values $x_i[0]$ are given. Until we reach Section~VI, we assume that no delay is present in the communication among nodes.

In this problem setting, the agents not only communicate with their direct neighbors as in conventional schemes,
but also with their multi-hop neighbors. Let $l$ be the maximum number of hops allowed in the network. Specifically, 
node $i_1$ can send messages of its own to an $l$-hop neighbor $i_{l+1}$ via different paths.
We represent a message as a tuple $m=(w,P)$, where $w=\mathrm{value}(m)\in \mathbb{R}$ is the message content and $P=\mathrm{path}(m)$ indicates the path via which message $m$ is transmitted. 
Moreover, nodes $i_1$ and $i_{l+1}$ are, respectively, the message source and the message destination.
When source node $i_1$ sends out a message, $P$ is a path vector of length $l+1$ with the source node being $i_1$ and other entries being empty. Then the one-hop neighbor $i_2$ receives this message from $i_1$, and it stores the value of node $i_1$ for consensus and relays the value of node $i_1$ to all the one-hop neighbors of $i_2$ with the second entry of $P$ being $i_2$ and other entries being unchanged. This relay procedure will continue until every entry of $P$ of this message is occupied, i.e., this message reaches node $i_{l+1}$. We denote by $\mathcal{V}(P)$ the set of nodes in $P$.

We now outline the message exchanges among the agents. 
At each time $k$, normal node $i$ conducts the following steps:

\noindent\textit{1. Transmit step:} Transmit message $m_{ij}[k]=(x_i[k],P_{ij}[k])$ over each $l$-hop path to node $j
\in\mathcal{N}_i^{l+}$.

\noindent\textit{2. Receive step:} Receive messages $m_{ji}[k]=(x_j[k],P_{ji}[k])$ from $j\in \mathcal{N}_i^{l-}$, whose destination is $i$. 
Let $\mathcal{M}_i[k]$ be the set of messages that node $i$ received in this step.

\noindent\textit{3. Update step:} Update the state $x_i[k]$ as
\begin{equation}
x_i[k+1]=g_i(\mathcal{M}_i[k]),  \label{updaterule}
\end{equation}
where $g_i(\cdot)$ is a real-valued function of the states received in this time step, to be defined later.

In the Transmit step and Receive step, nodes exchange messages with others that are up to $l$ hops away. Then in the Update step, node $i$ updates its state using the received values in $\mathcal{M}_i[k]$. Note that the adversary nodes may deviate from this specification as we describe in the next subsection.

In an agent network equipped with multi-hop communication, as the consensus update rule \eqref{updaterule}, we can extend the common one (e.g., \cite{olfati2007consensus}). Let $u_i[k]$ denote the control input for node $i$ at time $k$. Each node updates as
\begin{equation}
\begin{aligned}
x_i[k+1]&=x_i[k]+u_i[k], \\
u_i[k]&=-\sum_{j\in \mathcal{N}_i^{l-}} a_{ij}[k](x_i[k]-x_j[k]).
\end{aligned}
\end{equation}
This system can be given in the compact form as
\begin{equation}\label{m1}
\begin{aligned}
x[k+1]&=x[k] +  u[k],\\
u[k]&=-L[k]x[k],
\end{aligned}
\end{equation}
where $x[k]\in \mathbb{R}^n$ and $u[k]\in \mathbb{R}^n$ are the state vector and control input vector, respectively, and $L[k]$ is the Laplacian matrix of the $l$-th power of $\mathcal{G}$ determined by the messages $m_{ij}[k], i\in \mathcal{V} \ \text{and} \ j\in \mathcal{N}_i^{l-}$. 
As a generalization of the one-hop result (e.g. \cite{bullo2009distributed}, \cite{mesbahi2010graph}), it is obvious that with $l$-hop communication, consensus is possible if $\mathcal{G}^l$ has a rooted spanning tree.

\subsection{Threat Model}\label{threatmodel}

We introduce the threat model adopted in this paper. In the network, the node set $\mathcal{V}$ is partitioned into the set of normal nodes $\mathcal{N}$ and the set of adversary nodes $\mathcal{A}$. The latter set $\mathcal{A}$ is unknown to the normal nodes at all times. Moreover, we constrain the class of adversaries as follows (see, e.g., \cite{leblanc2013resilient}):
\begin{definition}
	\textit{($f$-total set)}
	The set of adversary nodes $\mathcal{A}$ is said to be $f$-total
	if it contains at most $f$ nodes, i.e., $\left| \mathcal{A}\right| \leq f$.
\end{definition}

\begin{definition}
	\textit{(Malicious nodes)}
	An adversary node $i\in \mathcal{A}$ is said to be a malicious node
	if it can arbitrarily modify its own value and relayed values, but sends the same state
	and relayed values to its neighbors at each iteration. 
	It can also decide not to send any value.\footnote[2]{This behavior corresponds to the omissive/crash model \cite{Lynch}.}
\end{definition}

As commonly done in the literature \cite{leblanc2013resilient}, \cite{su2017reaching}, 
we assume that each normal node knows the value of $f$ and the topology information of the graph up to $l$ hops. 
Moreover, the malicious model is reasonable in applications such as wireless sensor networks, where neighbors' information is obtained
by broadcast communication. 
In the multi-hop setting studied in this paper, it is important to impose the following assumption.

\begin{assumption}\label{assumptionpath}
	Each malicious node $i$ cannot manipulate the path values in the 
	messages containing its own state $x_i[k]$ and those that it relays. 
\end{assumption}

This is introduced for ease of analysis, but is not a strong constraint. In fact, manipulating message paths can be easily detected and hence does not create problems. We show how this can be done in Section~II-E.

\subsection{Resilient Asymptotic Consensus}

We now introduce the type of consensus among the normal agents to be sought in this paper \cite{leblanc2013resilient}, \cite{su2017reaching},  \cite{dibaji2017resilient}.

\begin{definition}
	If for any possible sets and behaviors of the
	malicious agents and any state values of the normal
	nodes, the following two conditions are satisfied,
	then we say that the normal agents reach 
	resilient asymptotic consensus:
	
	\begin{enumerate}
		\item Safety: There exists a bounded safety interval $\mathcal{S}$ determined by the initial values of the normal agents such that $x_i[k] \in \mathcal{S}, \forall i \in \mathcal{N}, k \in \mathbb{Z}_+$. 
		\item Agreement: There exists a state $x^*\in \mathcal{S}$
		such that $\lim_{k\to \infty}x_i[k]=x^*,  \forall i\in \mathcal{N}$.
	\end{enumerate}
	
\end{definition}

The problem studied in this paper is to develop an MSR-based algorithm for agents that can make $l$-hop communication to reach resilient consensus under the $f$-total malicious model and to characterize conditions on the network topology for the algorithm to properly perform. 
Note that in general, for MSR-based algorithms with one-hop communication, resilient consensus can be achieved under the $f$-total model with the necessary and sufficient condition expressed in terms of the so-called graph robustness; see, e.g., \cite{dibaji2018resilient, leblanc2013resilient}, and the following sections for the definition of robust graphs and related discussions.

\subsection{Discussion on Manipulation in Message Path Information}

It is notable that multi-hop communication is vulnerable to false data injection in the information relayed by nodes, which can make the problem of resilient consensus more complicated than the one-hop case. Earlier, Assumption \ref{assumptionpath} was introduced stating that the malicious nodes however cannot manipulate the path information in messages that they relay. We briefly explain here how such attacks can be detected, inspired by the discussion in \cite{su2017reaching}.

Such detection requires that each node can identify the neighbor from which it receives each message, which is commonly assumed (e.g., \cite{su2017reaching}, \cite{dolev1982byzantine}). 
Moreover, there are many methods to realize this function in real-world applications. For instance, by using the encryption technique of the RSA algorithm \cite{rivest1978method}, each node can send out its value associated with a digital signature using its own private key. Then, using the sender’s public key, the receiver can confirm that this message is indeed sent by the particular sender.


In each iteration of the synchronous algorithm, there are three potential cases where a message sent to normal node $i$ is
manipulated in its path information:
(i) Node $i$ receives multiple messages along the same path $P$; (ii) it receives messages along an unknown path $P'$; or
(iii) it does not receive any message along a known path $P$. Note that a normal node receives only one message along each path in each iteration when no adversarial node is present.

For case (i), this faulty behavior is caused either by duplicating messages or by manipulating path information in messages. 
We show that in both situations, the receiving node $i$ can find that there is at least one faulty node in path $P$. It is obvious for the first situation. For the path manipulating situation, consider the case where a normal node $h$ receives a message $m =(w, P)$ directly from node $j$ but the path $P$ does not contain node $j$. Then node $h$ knows that node $j$ is faulty, and will not forward the message. 
This indicates that in general, if there is a sequence of faulty nodes along a path, then the last one in the sequence must keep its own index within the path information in the messages that it transmits.
Moreover, this argument also holds for case (ii), i.e., node $i$ knows that at least one node in path $P'$ is faulty.
Therefore, in cases (i) and (ii), from the perspective of node $i$, manipulating the message path data is equivalent to having a faulty node in $P$ or $P'$ sending additional messages with manipulated values, and thus it will remove any values in this path by the MW-MSR algorithm. 

For case (iii), either a faulty node does not send/forward the message $m$, or it manipulates the message path $P$. In the latter case, for node $i$, manipulating the message path is equivalent to having a faulty node in $P$ not sending/forwarding the message.

Actually, this analysis can be extended to the algorithms with asynchronous updates.
In each update of such algorithms, consider the following three path manipulating cases for node $i$: (i) Node $i$ receives multiple messages along one path $P$ at the same time step; (ii) it receives messages along an unknown path $P'$; or
(iii) it does not receive any message along path $P$ in a period of time $\tau$, where $\tau$ is the maximum time delay of normal agents. Note that in case (i) for the asynchronous algorithm, faulty nodes can send multiple messages along $P$ as long as these faulty messages do not arrive at node $i$ at the same time step and this behavior will not affect normal agents, since only the most recent values of multi-hop neighbors will be used in the asynchronous MW-MSR algorithm. The analysis of cases (ii) and (iii) is similar to that of the synchronous algorithm.

The above analysis is based on the assumption that there is no packet loss in the fault-free networks. In real-world applications, packet losses can happen even in fault-free networks. We note that there are methods to deal with this issue.
Packet losses in the communication between two normal nodes can also cause the situation of case (iii) mentioned above. Like the one-hop W-MSR algorithm, if a packet loss happens in the communication from neighbor $j$ to node $i$, then node $i$ may receive only $|\mathcal{N}_i|-1$
values at this particular time step and still remove $f$ largest and $f$ smallest received values; hence, node $i$ uses less information from normal nodes to update. This behavior will not violate the safety interval, but it may slow down the speed of consensus. If the packet loss behavior happens frequently in this transmission path, then node $i$ can consider this path containing faulty nodes.

\section{Multi-hop Weighted MSR Algorithm}

In this section, we introduce the multi-hop weighted MSR (MW-MSR) algorithm, which is designed to solve the resilient consensus problem under the multi-hop setting. We first introduce the notion of message cover which plays a key role in the trimming function of our MSR algorithm. Then we outline the structure of the MW-MSR algorithm and provide examples to illustrate the idea behind the algorithm.

The notion of message cover \cite{su2017reaching} is crucial in the update rule of our algorithm to be proposed in this section. It evaluates the effects of adversary nodes that can possibly manipulate the updates of normal nodes in a multi-hop communication setting. Its formal definition is given as follows.

\begin{definition} For a graph $\mathcal{G} = (\mathcal{V},\mathcal{E})$, let $\mathcal{M}$ be a set of messages transmitted over $\mathcal{G}$, and let $\mathcal{P}(\mathcal{M})$ be the set of message paths of all the messages in $\mathcal{M}$, i.e., $\mathcal{P}(\mathcal{M}) =\{\mathrm{path}(m):m \in \mathcal{M}\}$. A \textit{message cover} of $\mathcal{M}$ is a set of nodes $\mathcal{T}(\mathcal{M})\subset \mathcal{V}$ whose removal disconnects all message paths, i.e., for each path $P\in \mathcal{P}(\mathcal{M})$, we have $\mathcal{V}(P)\cap \mathcal{T}(\mathcal{M})\neq \emptyset$. In particular, a \textit{minimum} message cover of $\mathcal{M}$ is defined by
	\begin{equation*}
	\mathcal{T}^*(\mathcal{M})\in	\arg \min_{\substack{ \mathcal{T}(\mathcal{M}): \textup{ Cover of } \mathcal{M}}} 	\left|  \mathcal{T} (\mathcal{M})\right| . 
	\end{equation*}
\end{definition}
\vspace{0.12cm}

As a simple example, consider the set $\mathcal{M}$ of paths connecting node $i$ to node $j$ which do not overlap. Then, its message cover must contain at least one node per path.
Clearly, there may be multiple minimum message covers if the paths are of length greater than three.

\begin{algorithm}[t] 
	\caption{MW-MSR Algorithm} 
	
	1) At each time $k$, normal node $i$ 
	sends its own message to nodes in $\mathcal{N}_i^{l+}$.
	Then, it obtains messages of the nodes in $\mathcal{N}_i^{l-}$ and itself, whose set is denoted by $\mathcal{M}_i[k]$, and sorts the values in $\mathcal{M}_i[k]$ in an increasing order.
	
	2) (a) Define two subsets of $\mathcal{M}_i[k]$ based on the message values:
	\begin{equation*}
	\overline{\mathcal{M}}_i[k]=\{ m\in \mathcal{M}_i[k]: \mathrm{value}(m)> x_i[k]  \},
	\end{equation*}
	\begin{equation*}
	\underline{\mathcal{M}}_i[k]=\{ m\in \mathcal{M}_i[k]: \mathrm{value}(m)< x_i[k]  \}.
	\end{equation*}
	
	(b) Then, let $\overline{\mathcal{R}}_i[k]=\overline{\mathcal{M}}_i[k]$ if the cardinality of a minimum cover of $\overline{\mathcal{M}}_i[k]$ is less than $f$, i.e., $\left|  \mathcal{T}^* (\overline{\mathcal{M}}_i[k])\right| <f$. Otherwise, let $\overline{\mathcal{R}}_i[k]$ be the largest sized subset of $\overline{\mathcal{M}}_i[k]$ such that (i) for all $m\in \overline{\mathcal{M}}_i[k]\setminus \overline{\mathcal{R}}_i[k]$ and $m'\in \overline{\mathcal{R}}_i[k]$ we have $\mathrm{value}(m) \leq \mathrm{value}(m')$, and (ii) the cardinality of a minimum cover of $\overline{\mathcal{R}}_i[k]$ is exactly $f$, i.e., $\left|  \mathcal{T}^* (\overline{\mathcal{R}}_i[k])\right| =f$. 
	
	(c) Similarly, let $\underline{\mathcal{R}}_i[k]=\underline{\mathcal{M}}_i[k]$ if the cardinality of a minimum cover of $\underline{\mathcal{M}}_i[k]$ is less than $f$, i.e., $\left|  \mathcal{T}^* (\underline{\mathcal{M}}_i[k])\right| <f$. Otherwise, let $\underline{\mathcal{R}}_i[k]$ be the largest sized subset of $\underline{\mathcal{M}}_i[k]$ such that (i) for all $m\in \underline{\mathcal{M}}_i[k]\setminus \underline{\mathcal{R}}_i[k]$ and $m'\in \underline{\mathcal{R}}_i[k]$ we have $\mathrm{value}(m) \geq \mathrm{value}(m')$, and (ii) the cardinality of a minimum cover of $\underline{\mathcal{R}}_i[k]$ is exactly $f$, i.e., $\left|  \mathcal{T}^* (\underline{\mathcal{R}}_i[k])\right| =f$. 
	
	(d) Finally, let $\mathcal{R}_i[k]=\overline{\mathcal{R}}_i[k]\cup\underline{\mathcal{R}}_i[k]$.

	3) Node $i$ updates its value as follows:
	\begin{equation}
	x_i[k+1]=\sum_{m\in \mathcal{M}_i[k]\setminus \mathcal{R}_i[k]} a_{i}[k]\mathrm{value}(m),  \label{msrupdate}
	\end{equation}
	where $a_{i}[k]=1/(\left| \mathcal{M}_i[k]\setminus \mathcal{R}_i[k] \right| )$.
\end{algorithm}

Now, we are ready to introduce the structure of the synchronous \textit{MW-MSR algorithm} in Algorithm 1.
Note that the one-hop version of the MW-MSR algorithm (i.e., with $l=1$) is equivalent to the W-MSR algorithm in \cite{leblanc2013resilient}.
However, we can see that the difference between the MW-MSR algorithm and the W-MSR algorithm mainly lies in the trimming function in step 2 when $l\geq 2$. 
For general MSR algorithms of one-hop communication, the essential idea for the normal nodes to avoid being affected by adversary nodes is that 
each normal node $i$ will exclude the effects from $f$ neighbors with extreme values (possibly sent by faulty nodes).
This can guarantee that values outside the safety interval will not be used by any normal node at any time step. In the one-hop case, for each normal node $i$, the number of values received from such $f$ neighbors is exactly $f$, i.e., node $i$ will trim away $f$ largest and $f$ smallest values at each step. This is because each neighbor sends only one value of its own to node $i$ at each step under the typical assumptions made in MSR-related works \cite{leblanc2013resilient}, \cite{vaidya2012iterative}.

Under the multi-hop setting, the situation changes significantly even if we assume that each node can only send out one value of its own to its neighbors at each step. Since each node relays the values from different neighbors, normal node $i$ can receive more than one value from one direct neighbor. Thus, in the MW-MSR algorithm, normal node $i$ cannot just trim away $f$ largest and $f$ smallest values at each step. Instead, it needs to trim away the largest and smallest values from exactly $f$ nodes within $l$ hops away, which is the generalization of the essential idea in the one-hop W-MSR algorithm.

\begin{figure}[t]
	\centering
	\subfigure[]{
		\includegraphics[width=1.65in]{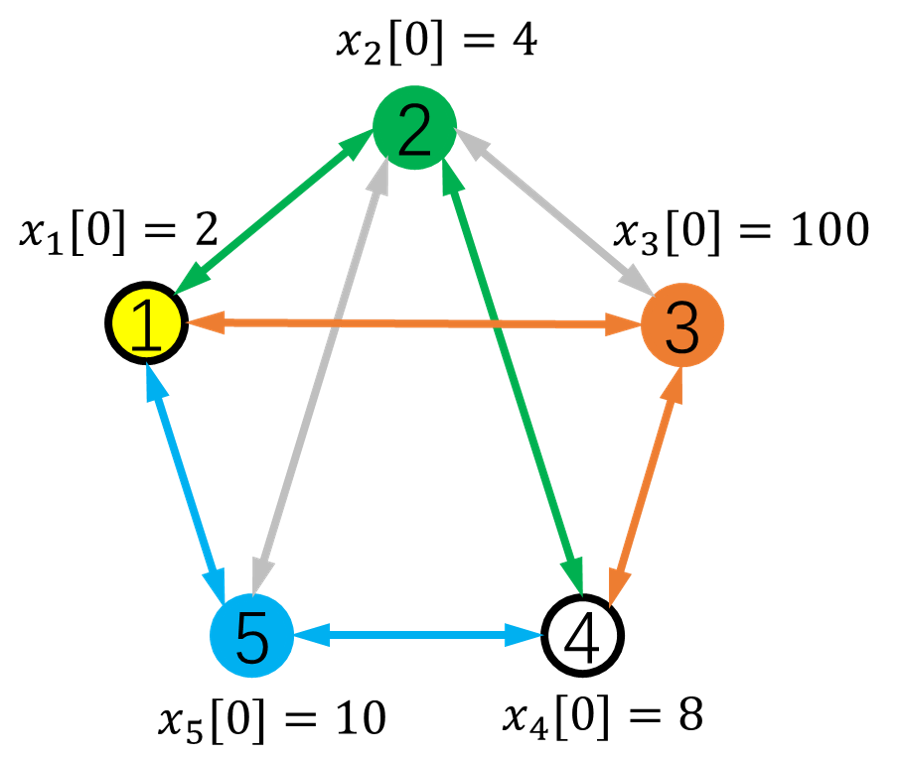}
	}
	\quad
	\vspace{-5pt}
	\subfigure[]{
		\includegraphics[width=1.1in]{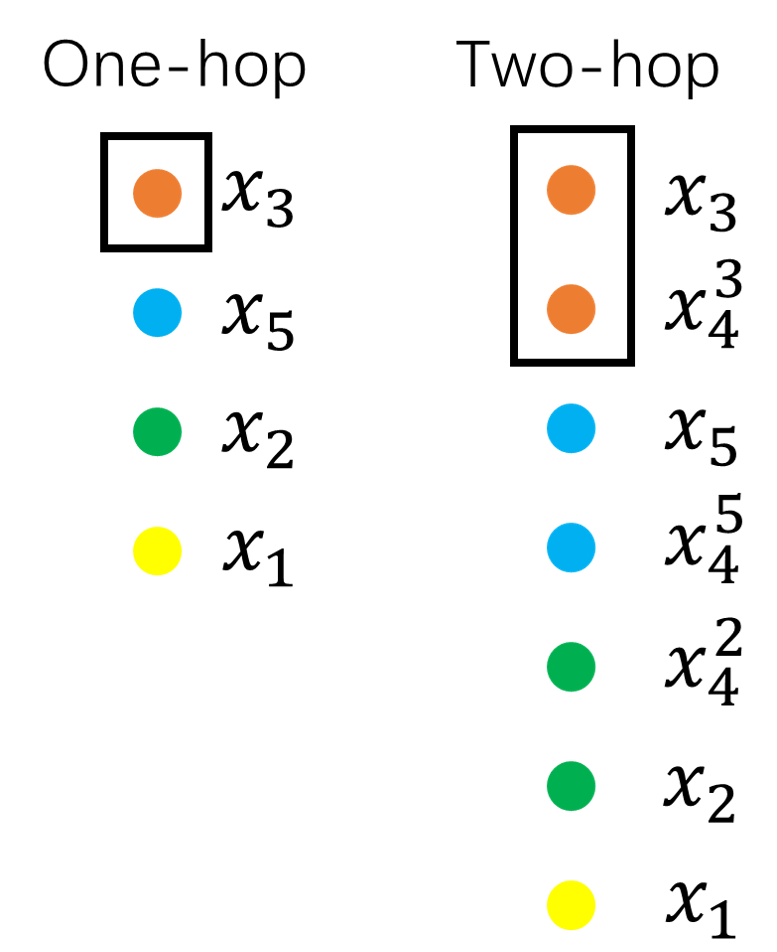}
	}
	\vspace{-4pt}
	\caption{(a) A 5-node graph. (b) Values removed by node 1 in one-hop and two-hop algorithms.}
	\label{one-hop-two-hop}
\end{figure}

	To characterize the number of the extreme values from exactly $f$ nodes for node $i$, the notion of minimum message cover (MMC) is designed. 
	Intuitively speaking, for normal node $i$, $\overline{\mathcal{R}}_i[k]$ and $\underline{\mathcal{R}}_i[k]$ are the largest sized sets of received messages containing very large and small values that may have been generated or tampered by $f$ adversary nodes, respectively. 
	Here, we focus on how $\underline{\mathcal{R}}_i[k]$ is determined, as $\overline{\mathcal{R}}_i[k]$ can be obtained in a similar way.
	When the cardinality of the MMC of set $\underline{\mathcal{M}}_i[k]$ is no more than $f$, node $i$ simply takes $\underline{\mathcal{R}}_i[k]=\underline{\mathcal{M}}_i[k]$. 
	Otherwise, node $i$ will check the first $f+1$ values of $\underline{\mathcal{M}}_i[k]$, and if the MMC of these values is of cardinality $f$, then it will check the first $f+2$ values of $\underline{\mathcal{M}}_i[k]$. This procedure will continue until for the first $f+h$ ($h\geq 1$) values of $\underline{\mathcal{M}}_i[k]$, the MMC of these values is of cardinality $f+1$. Then $\underline{\mathcal{R}}_i[k]$ is taken as the first $f+h-1$ values of $\underline{\mathcal{M}}_i[k]$. 
	After sets $\overline{\mathcal{R}}_i[k]$ and $\underline{\mathcal{R}}_i[k]$ are determined, in the control input $u_i(k)$ computed by \eqref{msrupdate} in step 3, values in these sets are excluded. Note that this control is consistent with the one in \eqref{updaterule} when $f=0$.

We also illustrate the determination of such subsets through a simple example.
Consider the network in Fig. \ref{one-hop-two-hop} with initial states $x[0]=[2\ 4\ 100\ 8\ 10]^T$, where node 3 is set to be malicious ($f=1$) and constantly broadcasts the value 100 as its own value as well as those in the relayed messages. We look at node 1 at time $k=0$ and drop the time index $k$. In the one-hop version of the MW-MSR algorithm, the input for node 1 is $\{x_1, x_2, x_5, x_3\}$, and it chooses $\overline{\mathcal{R}}_1[0]=\{x_3=100\}$ and $\underline{\mathcal{R}}_1[0]=\emptyset$ in step 2 of the algorithm (since the value $x_1$ is the smallest in the input). 

In the two-hop version of the MW-MSR algorithm, node 1 receives the state values $x_2, x_3,$ and $ x_5$ directly from nodes 2, 3, and 5, respectively. Moreover, it receives the relayed values of node 4 through nodes 2, 3, and 5, denoted by $x_4^2, x_4^3,$ and $ x_4^5$. Then 
the sorted input for node 1 is $\{x_1, x_2, x_4^2,  x_4^5, x_5, x_4^3, x_3\}$, and node 1 checks the MMC of the subset of the largest values starting from the $(f+1)$th value (since the values before the $f$th one are definitely removed by node 1). First, it evaluates $\{x_4^3, x_3\}$, and the MMC of this message set is the node set $\{3\}$ with cardinality 1. Then, it evaluates $\{x_5, x_4^3, x_3\}$ and the MMC of this message set can be found to be the node set $\{3, 5\}$ with cardinality 2, which is bigger than $f=1$. As a result, node 1 confirms that $\{x_4^3, x_3\}$ is the largest sized set of the large values that may have been generated or tampered by $f$ adversary nodes. Therefore, node 1 chooses $\overline{\mathcal{R}}_1[0]=\{x_4^3=100, x_3=100\}$ and $\underline{\mathcal{R}}_1[0]=\emptyset$ in step 2 of the algorithm.

In this paper, the key question to be addressed is, under what conditions on the network can the above algorithm achieve resilient asymptotic consensus? Our approach is to develop a generalization of the results and analysis of the one-hop case. In particular, this necessitates us to extend the notion of graph robustness by taking account of multi-hop communication. This is carried out in the next section.

\section{Robustness with Multi-hop Communication}

In this section, we discuss the notion of graph robustness.
This notion was first introduced in \cite{leblanc2013resilient}, which corresponds
to the one-hop case. We provide its multi-hop generalization,
which plays a crucial role in our resilient consensus problem. 

As in the definition of robustness for the one-hop case \cite{leblanc2013resilient}, 
we start with the definition of $r$-reachability. Specifically, 
when the communication is only one-hop, 
a node set is said to be $r$-reachable if it contains
at least one node that has at least $r$ incoming neighbors outside this set. This notion basically captures the capability of a set to be influenced by the outside of the set when the nodes apply the MSR algorithms with parameter $r-1$.


\begin{figure}[t]
	\centering
	\subfigure[]{
		\includegraphics[width=1.25in]{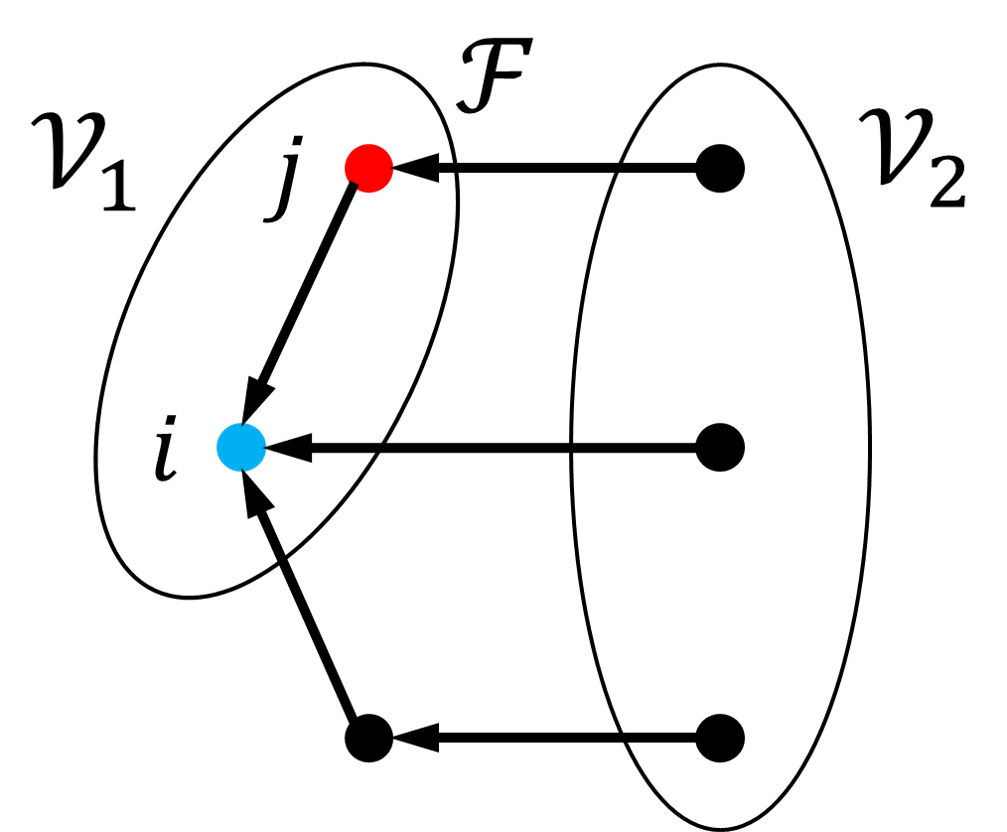}
		
	}
	\quad
	\vspace{-3pt}
	\subfigure[]{
		\includegraphics[width=1.25in]{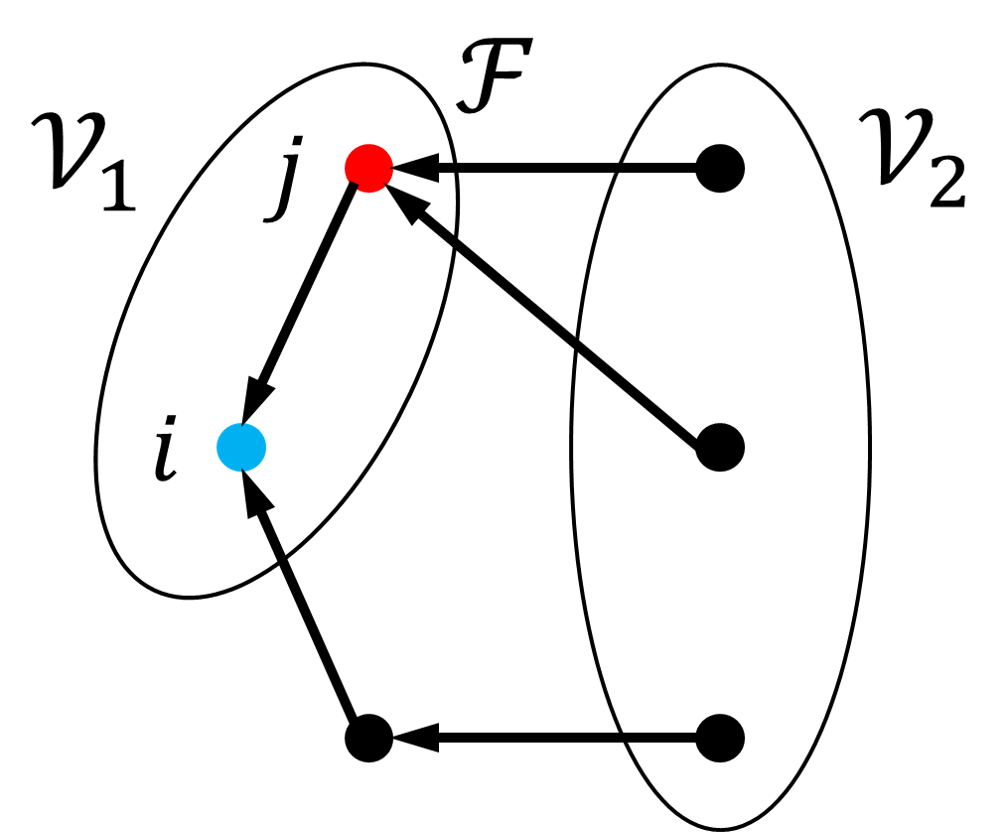}
		
	}
	\vspace{-3pt}
	\caption{(a) Node $i$ has two independent paths originating from the outside of $\mathcal{V}_1$ with respect to set $\mathcal{F}=\{j\}$. (b) Node $i$ only has one independent path sharing the same property.}
	\label{3-reachable}
	\vspace*{-3.5mm}
\end{figure}

In generalizing this notion to the case of multi-hop communication, it is crucial to extend the above-mentioned capability. In particular, the influence from the outside of the set may come from remote nodes and are not restricted to direct neighbors. With a slight change, in the multi-hop setting, we define the $r$-reachability as follows.

\begin{definition}
	Consider a graph $\mathcal{G} = (\mathcal{V},\mathcal{E})$ with $l$-hop communication. For $r\in \mathbb{Z}_+$, set $\mathcal{F}\subset \mathcal{V}$, and nonempty set $\mathcal{V}_1\subset \mathcal{V}$, a node $i\in \mathcal{V}_1$
	is said to be $r$-reachable with $l$-hop communication with respect to $\mathcal{F}$ if it has at least $r$ independent paths (i.e., only node $i$ is the common node in these paths) of at most $l$ hops originating from nodes outside $\mathcal{V}_1$ and all these paths do not have any node in set $\mathcal{F}$ as an intermediate node (i.e., the nodes in $\mathcal{F}$ can be source or destination nodes in these paths).
\end{definition}

Intuitively speaking, for any set $\mathcal{F}\subset \mathcal{V}$ and for node $i\in \mathcal{V}_\text{1}$ to have the above-mentioned property, there should be at least $r$ source nodes outside $\mathcal{V}_\text{1}$ and at least one independent path of length at most $l$ hops from each of the $r$ source nodes to node $i$, where such a path does not contain any internal node from the set $\mathcal{F}$.
It is clear that for the one-hop case, to count the independent paths simply becomes to count the in-neighbors.

As an example, consider the graph in Fig.~\ref{3-reachable}(a), where the two node sets $\mathcal{V}_1$ and $\mathcal{V}_2$ are taken as indicated and the set $\mathcal{F}=\{j\}$. Here, node $i\in \mathcal{V}_1$ has two independent paths of at most two hops originating from the nodes outside $\mathcal{V}_1$ with respect to the set $\mathcal{F}$. In contrast, in a similar graph shown in Fig.~\ref{3-reachable}(b), such a property is lost and node $i$ has only one path from the outside of $\mathcal{V}_1$ w.r.t. the set $\mathcal{F}$.


\begin{figure}[t]
	\centering
	\subfigure[]{
		\includegraphics[width=0.85in]{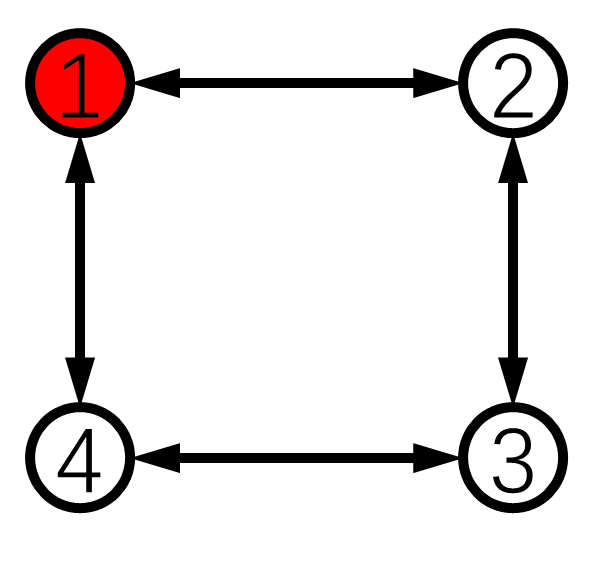}
		
	}
	\quad
	\subfigure[]{
		
		\includegraphics[width=1.35in]{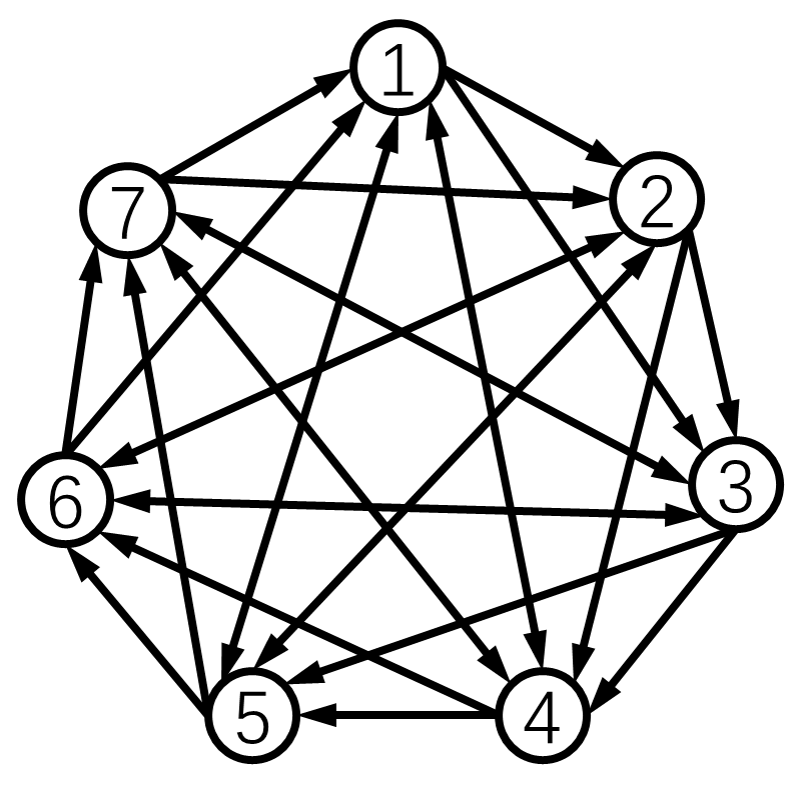}
	}
	\vspace{-3pt}
	\caption{(a) The graph is not $(2, 2)$-robust with one hop, but it is $(2,2)$-robust with $2$ hops. (b) The graph is $(2,2)$-robust with one hop and $(3,3)$-robust with $2$ hops.}
	\label{graph1}
	\vspace*{-3.5mm}
\end{figure}

Now, we are ready to generalize this notion to the entire graph and define $r$-robustness and $(r,s)$-robustness with $l$ hops as follows.


\begin{definition}\label{rs-robust} A directed graph $\mathcal{G} = (\mathcal{V},\mathcal{E})$ is said to be $(r,s)$-robust with $l$ hops with respect to a given set $\mathcal{F}\subset \mathcal{V}$,
	if for every pair of nonempty disjoint subsets $\mathcal{V}_\text{1},\mathcal{V}_\text{2}\subset \mathcal{V}$, at least one of the following conditions holds:
	
	
	\begin{enumerate}
		\item $\mathcal{Z}_{\mathcal{V}_1}^r=\mathcal{V}_1$,
		\item $\mathcal{Z}_{\mathcal{V}_2}^r=\mathcal{V}_2$,
		\item  $\left| \mathcal{Z}_{\mathcal{V}_1}^r\right| +\left| \mathcal{Z}_{\mathcal{V}_2}^r\right| \geq s$,
	\end{enumerate}
	
	\noindent where $\mathcal{Z}_{\mathcal{V}_a}^\textit{r}$ is the set of nodes in $\mathcal{V}_\textit{a}$ ($a=1,2$) that are $r$-reachable with $l$-hop communication with respect to $\mathcal{F}$.
	Moreover, if the graph $\mathcal{G}$ satisfies this property with respect to any set $\mathcal{F}$ following the $f$-total model, then we say that $\mathcal{G}$ is $(r,s)$-robust with $l$ hops under the $f$-total model. 
	When it is clear from the context, we will just say $\mathcal{G}$ is $(r,s)$-robust with $l$ hops.
	Furthermore, when the graph is $(r,1)$-robust with $l$ hops, we also say it is $r$-robust with $l$ hops.
\end{definition}

Generally, robustness of a graph increases as the relay range $l$ increases.
We will illustrate this point using the graphs in Fig.~\ref{graph1}. 
Note that graph robustness with multi-hop communication needs to be checked for every possible set $\mathcal{F}$ satisfying the $f$-total model. In the context of this paper, we are interested to check $(r,s)$-robustness under $(r-1)$-total model.

First, the graph in Fig.~\ref{graph1}(a) is not $(2,2)$-robust with one hop since for the sets $\{1,2\}$ and $\{3,4\}$, none of the nodes has two in-neighbors outside the corresponding set. However, under the 1-total model, this graph is $(2,2)$-robust with $2$ hops. For instance, we first check the condition for the set $\mathcal{F}=\{1\}$. For sets $\{1,2\}$ and $\{3,4\}$, all of the nodes 1, 3 and 4 have two independent paths of at most two hops originating from the outside of the set with node 1 not being the internal node. After checking all the possible subsets of $\mathcal{V}$, one can confirm that this graph is $(2,2)$-robust with $2$ hops with respect to set $\mathcal{F}=\{1\}$. Since this graph is actually symmetric for each node, we can conclude that for the set $\mathcal{F}=\{v\}$  ($v=2,3,4$), this graph is $(2,2)$-robust with $2$ hops with respect to this set. Hence, this graph is $(2,2)$-robust with $2$ hops.

Next, we look at the graph in Fig.~\ref{graph1}(b).
When $l=1$, this graph is $(2,2)$-robust with $1$ hop but not $(3,3)$-robust with $1$ hop.
When $l=2$, it becomes $(3,3)$-robust with $2$ hops. 
It is further noted that the level of robustness is constrained
by the in-degrees of the nodes. 
In the graph of Fig.~\ref{graph1}(b), each node has four incoming edges. As a result,
for $r\geq 4$, this graph cannot be $(r, s)$-robust with any number of hops. 
We elaborate more on this aspect in Section~VII.

\section{Synchronous Network}

In this section, we analyze the MW-MSR algorithm under synchronous updates, i.e., each normal node will update its value using those received from all of its $l$-hop neighbors in a synchronous manner with other nodes at each time $k$.

\subsection{Matrix Representation}
First, we write the system in a matrix form.
For ease of notation in our analysis, reorder the node indices so that the normal nodes take indices $1,\dots,n_N$ and the malicious nodes are $n_N+1,\dots,n$. Then the state vector and control input vector can be written as 
\begin{equation} 
x[k]=\begin{bmatrix} x^N[k]  \\ x^F[k] \end{bmatrix}, \medspace u[k]=\begin{bmatrix} u^N[k]  \\ u^F[k] \end{bmatrix}, 
\end{equation}
where the superscript $N$ stands for normal and $F$ for faulty. Regarding the control inputs $u^N[k]$ and $u^F[k]$, the normal nodes follow \eqref{msrupdate} while the malicious nodes may not. Hence, they can be expressed as
\begin{equation}
\begin{array}{lll} 
u^N[k] =-L^N[k]x[k],\\
u^F[k] : \textup{arbitrary,}
\end{array}
\end{equation}
where $L^N[k]\in \mathbb{R}^{n_N\times n}$ is the matrix formed by the first $n_N$ rows of $L[k]$ associated with normal nodes. The row sums of this matrix $L^N[k]$ are zero as in $L[k]$.
Thus, we can rewrite the system as
\begin{equation} \label{system1}
x[k+1]=\left(  I_n -   \begin{bmatrix} L^N[k]  \\ 0 \end{bmatrix} \right)  x[k] +   \begin{bmatrix} 0  \\ I_{n_F} \end{bmatrix}u^F[k].
\end{equation}

\vspace{0.12cm}

\subsection{Consensus Analysis with Multi-hop Communication}

Now we are ready to provide a necessary and sufficient condition for resilient consensus applying the synchronous MW-MSR algorithm. The following theorem is the first main contribution of this paper. 

\begin{theorem}\label{syn}
	Consider a directed graph $\mathcal{G} = (\mathcal{V},\mathcal{E})$ with $l$-hop communication, where each normal node updates its value according to the synchronous MW-MSR algorithm with parameter
	$f$. Under the $f$-total malicious model, resilient asymptotic
	consensus is achieved if and only if the network topology is
	$(f + 1, f + 1)$-robust with $l$ hops.
	Moreover, the safety interval is given by 
	\begin{equation*}  
	\mathcal{S}=\big[ \min x^N[0], \max x^N[0] \big].        	
	\end{equation*}
\end{theorem}
\vspace{0.12cm}

\begin{proof}
	\textit{(Necessity)} If $\mathcal{G}$ is not $(f +1, f +1)$-robust with $l$ hops, then
	there are nonempty, disjoint subsets $\mathcal{V}_1, \mathcal{V}_2\subset \mathcal{V}$ such that none of
	the conditions in Definition \ref{rs-robust} holds. Suppose that the initial value of
	each node in $\mathcal{V}_1$ is $a$ and each node in $\mathcal{V}_2$ takes $b$, with $a < b$.
	Let all other nodes have initial values taken from the interval
	$(a, b)$. Since $| \mathcal{Z}_{\mathcal{V}_1}^{f+1}| +| \mathcal{Z}_{\mathcal{V}_2}^{f+1}| \leq f$, suppose that all nodes in $\mathcal{Z}_{\mathcal{V}_1}^{f+1}$ and $\mathcal{Z}_{\mathcal{V}_2}^{f+1}$ are malicious and take constant values.
	Then there is still at least one
	normal node in both $\mathcal{V}_1$ and $\mathcal{V}_2$ since $| \mathcal{Z}_{\mathcal{V}_1}^{f+1}| < \left| \mathcal{V}_1\right|$ and $| \mathcal{Z}_{\mathcal{V}_2}^{f+1}| < \left| \mathcal{V}_2\right|$, respectively. Then these normal nodes remove all the values of incoming neighbors outside of their respective sets since the message cover of these values has cardinality equal to $f$ or less. According to the synchronous MW-MSR algorithm, such normal nodes will keep their values and consensus cannot be achieved.
	
	\textit{(Sufficiency)} 
	First, we show that the safety condition of resilient consensus is satisfied. 
	Let $\overline{x}^N[k]$ and $\underline{x}^N[k]$ to be
	the maximum and minimum values of the normal nodes at
	time $k$, respectively. 
	We can show that $\overline{x}^N[k]$ is monotonically nonincreasing and $\underline{x}^N[k]$ is monotonically nondecreasing, and thus each of them has some limit.
	This can be directly shown from the definitions and the facts that the values 
	used in the MW-MSR update rule always lie within the interval $\big[ \underline{x}^N[k], \overline{x}^N[k] \big] \subseteq \mathcal{S}$ for $k\geq 0$.
	Since at each time $k$, in step 2 of Algorithm 1, node $i$ wipes out the possibly manipulated values from at most $f$ nodes within $l$ hops.
	Moreover, the update rule \eqref{system1} uses a convex combination of the values in $\big[ \underline{x}^N[k], \overline{x}^N[k] \big]$. 
	Therefore, the safety condition is satisfied.

	Then, we denote the limits of $\overline{x}^N[k]$ and $\underline{x}^N[k]$ by $\overline{\omega}$ and $\underline{\omega}$,
	respectively. We will prove by contradiction to show that $\overline{\omega}=\underline{\omega}$, and thus the normal nodes will reach consensus.
	Suppose that $\overline{\omega} > \underline{\omega}$. We can then take $\epsilon_0 > 0$ such that $\overline{\omega}-\epsilon_0 > \underline{\omega}+\epsilon_0$. 
	Fix $\epsilon<\epsilon_0\alpha^{n_N}/(1-\alpha^{n_N})$, where $0<\epsilon<\epsilon_0$ and $\alpha$ is the minimum of all $a_{i}[k]$ in step 3 of the MW-MSR algorithm.
	For $1\leq \gamma \leq n_N$, define $\epsilon_\gamma$ recursively as 
	\begin{equation*}
	\epsilon_{\gamma}= \alpha\epsilon_{\gamma-1}-(1-\alpha)\epsilon.
	\end{equation*}
	So we have $0 < \epsilon_{\gamma} < \epsilon_{\gamma-1}\leq \epsilon_0$ for all $\gamma$, since it holds that
	\begin{equation}\label{epsilon_positive}
	\begin{aligned}
	\epsilon_\gamma &= \alpha\epsilon_{\gamma-1}-(1-\alpha)\epsilon
	= \alpha^\gamma\epsilon_0-(1-\alpha^\gamma)\epsilon\\
	&\geq \alpha^{n_N}\epsilon_0-(1-\alpha^{n_N})\epsilon>0. 
	\end{aligned} 
	\end{equation}
	
	At any time step $k$ and for any $\epsilon_t>0$, define two sets: 
	\begin{equation*}
	\begin{aligned}
	\mathcal{Z}_1(k,\epsilon_t)&=\{i\in \mathcal{V}: x_i[k]>\overline{\omega}-\epsilon_t\},\\
	\mathcal{Z}_2(k,\epsilon_t)&=\{i\in \mathcal{V}: x_i[k]<\underline{\omega}+\epsilon_t\}.
	\end{aligned}
	\end{equation*}
	By the definition of $\epsilon_0$, $\mathcal{Z}_1(k,\epsilon_0)$ and $\mathcal{Z}_2(k,\epsilon_0)$ are disjoint.
	
	Let $k_\epsilon$ be the time such
	that $\overline{x}^N[k] < \overline{\omega}+\epsilon $ and $\underline{x}^N[k] > \underline{\omega}-\epsilon$, $\forall k\geq k_\epsilon$. Such a $k_\epsilon$ exists since $\overline{x}^N[k]$ and $\underline{x}^N[k]$ converge to $\overline{\omega}$ and $\underline{\omega}$, respectively, in monotonic manners as discussed above.
	Consider the nonempty and disjoint sets $\mathcal{Z}_1(k_\epsilon,\epsilon_0)$ and $\mathcal{Z}_2(k_\epsilon,\epsilon_0)$. Notice that
	the network is $(f + 1, f + 1)$-robust with $l$ hops w.r.t. any set $\mathcal{F}$ following the $f$-total model and the set of malicious nodes $\mathcal{A}$ also satisfies the $f$-total model. Hence, the network is $(f + 1, f + 1)$-robust with $l$ hops w.r.t. the set $\mathcal{A}$ and at least one of the three conditions in Definition \ref{rs-robust} holds.
	Also notice that the normal node with value $\overline{x}^N[k_\epsilon]$ will definitely be in set $\mathcal{Z}_1(k_\epsilon,\epsilon_0)$, and it is similar for the case of $\mathcal{Z}_2(k_\epsilon,\epsilon_0)$.
	Hence, all nodes in either $\mathcal{Z}_1(k_\epsilon,\epsilon_0)$ or $\mathcal{Z}_2(k_\epsilon,\epsilon_0)$ have the $(f+1)$-reachable property, or the union of the two sets contains at least $f + 1$ nodes having the $(f+1)$-reachable property. Since there are at most $f$ malicious nodes,
	for all cases, there must exist a normal node in the union of $\mathcal{Z}_1(k_\epsilon,\epsilon_0)$ and $\mathcal{Z}_2(k_\epsilon,\epsilon_0)$ such that it has at least $f +1$ independent paths originating from different nodes outside of its set and these paths do not have any internal node in $\mathcal{A}$.

		Suppose that normal
		node $i\in \mathcal{Z}_1(k_\epsilon,\epsilon_0)\cap \mathcal{N}$ has the $(f+1)$-reachable property. Thus, node $i$ has at least $f+1$ neighbors within $l$ hops outside set $\mathcal{Z}_1(k_\epsilon,\epsilon_0)$, i.e., the values of these neighbors are smaller than $x_i[k_\epsilon]$ and are at
		most equal to $\overline{\omega}- \epsilon_0$. Moreover,
		the original values of these multi-hop neighbors of node $i$ will definitely reach node $i$ even if the source nodes are malicious (since the internal nodes of these paths are all normal and they relay the values as received, without making any changes). 
		Hence, node $i$ will use at least one of these values to update its own. This is because in step 2(c) of Algorithm 1, 
		node $i$ will remove values lower than its own
		value of which the cardinality of the minimum message cover
		is at most $f$. As a result, among the neighbors within $l$ hops outside set $\mathcal{Z}_1(k_\epsilon,\epsilon_0)$, the values from up to $f$ of them will be disregarded by node $i$.

	Now, in Algorithm 1, the update rule \eqref{msrupdate} of step 3 is applied. Here, each coefficient of the neighbors is lower bounded by $\alpha$. 
	Since the largest value that node $i$ will use at time $k_\epsilon$ is $\overline{x}^N[k_\epsilon]$, placing the largest possible weight on $\overline{x}^N[k_\epsilon]$ produces
	\begin{equation*}
	\begin{aligned}
	&x_i[k_\epsilon+1] \leq (1-\alpha)\overline{x}^N[k_\epsilon]+\alpha(\overline{\omega}-\epsilon_0)\\
	&~~\leq (1-\alpha)(\overline{\omega}+\epsilon)+\alpha(\overline{\omega}-\epsilon_0) \leq \overline{\omega}-\alpha\epsilon_0+(1-\alpha)\epsilon.
	\end{aligned}
	\end{equation*}
	Note that this upper bound also applies to the updated value
	of any normal node not in $\mathcal{Z}_1(k_\epsilon,\epsilon_0)$, because such
	a node will use its own value in its update. Similarly, if node
	$i\in \mathcal{Z}_2(k_\epsilon,\epsilon_0)\cap \mathcal{N}$ has the $(f+1)$-reachable property, then $x_i[k_\epsilon+1]\geq \underline{\omega}+\alpha\epsilon_0-(1-\alpha)\epsilon$. Again, any normal node not in $\mathcal{Z}_2(k_\epsilon,\epsilon_0)$ will have the same lower bound.
	
	
	Next, consider the sets $\mathcal{Z}_1(k_\epsilon+1,\epsilon_1)$ and $\mathcal{Z}_2(k_\epsilon+1,\epsilon_1)$. By $\epsilon_1<\epsilon_0$, these two sets are still disjoint. Since at least one of the normal nodes in $\mathcal{Z}_1(k_\epsilon,\epsilon_0)$ decreases at
	least to $\overline{\omega}-\epsilon_1$ (or below), or one of the nodes in $\mathcal{Z}_2(k_\epsilon,\epsilon_0)$
	increases at least to $\underline{\omega}+\epsilon_1$ (or above), it must hold
	$\left| \mathcal{Z}_1(k_\epsilon+1,\epsilon_1)\cap \mathcal{N}\right|< \left| \mathcal{Z}_1(k_\epsilon,\epsilon_0)\cap \mathcal{N}\right|$, $\left| \mathcal{Z}_2(k_\epsilon+1,\epsilon_1)\cap \mathcal{N}\right|< \left| \mathcal{Z}_2(k_\epsilon,\epsilon_0)\cap \mathcal{N}\right|$, or both. Recall that $0 < \epsilon_{\gamma} < \epsilon_{\gamma-1}\leq \epsilon_0$. As long as there are still normal nodes in 
	$\mathcal{Z}_1(k_\epsilon+\gamma,\epsilon_\gamma)$ and/or $\mathcal{Z}_2(k_\epsilon+\gamma,\epsilon_\gamma)$, we can repeat the above analysis for time step $k_\epsilon+\gamma$, which will result in either $\left| \mathcal{Z}_1(k_\epsilon+\gamma,\epsilon_\gamma)\cap \mathcal{N}\right|< \left| \mathcal{Z}_1(k_\epsilon+\gamma-1,\epsilon_{\gamma-1})\cap \mathcal{N}\right|$, $\left| \mathcal{Z}_2(k_\epsilon+\gamma,\epsilon_\gamma)\cap \mathcal{N}\right|< \left| \mathcal{Z}_2(k_\epsilon+\gamma-1,\epsilon_{\gamma-1})\cap \mathcal{N}\right|$, or both. 
	
	Since $\left| \mathcal{Z}_1(k_\epsilon,\epsilon_0)\cap \mathcal{N}\right|+\left| \mathcal{Z}_2(k_\epsilon,\epsilon_0)\cap \mathcal{N}\right| \leq {n_N}$,
	there must be some time step $k_\epsilon + T$ (with $T\leq {n_N}$) such that either 
	$\mathcal{Z}_1(k_\epsilon+ T,\epsilon_T)\cap \mathcal{N} $ or $\mathcal{Z}_2(k_\epsilon+ T,\epsilon_T)\cap \mathcal{N} $ is empty. In the former case, all normal nodes in the network at time step
	$k_\epsilon+T$ have values at most $\overline{\omega}-\epsilon_T$, while in the latter case all
	normal nodes at time step $k_\epsilon + T$ have values
	no less than $\underline{\omega}+\epsilon_T$. 
	By \eqref{epsilon_positive} and $T\leq n_N$, it holds that $\epsilon_T > 0$.
	Hence, we have contradiction to the fact that the largest value monotonically
	converges to $\overline{\omega}$ (in the former case) or that the smallest
	value monotonically converges to $\underline{\omega}$ (in the latter case). Hence, it must be the case that $\epsilon_0 = 0$, proving that $\overline{\omega} = \underline{\omega}$.
\end{proof}

	We emphasize that the graph condition based on the notion of robustness with $l$ hops is tight for our MW-MSR algorithm. Our notion captures the capability of agents to be influenced by the outside of the set in the multi-hop settings.
	We note that in \cite{su2017reaching}, an idea similar to robustness is proposed, and based on it, a tight necessary and sufficient condition for Byzantine consensus using an MSR-type algorithm with multi-hop communication is provided. However, the focus there is on the Byzantine model and the condition is expressed in terms of
	the subgraph consisting of only the normal nodes.
	Part of the reason to focus on only the subgraph of normal nodes is that Byzantine nodes are more adversarial compared to malicious nodes as they can send different values to different neighbors. Hence, the subgraph of normal nodes has to be sufficiently robust to fight against the possible attacks.
	The condition there is an extension of the one for the one-hop case shown in \cite{vaidya2012iterative}. Moreover, to meet the condition there, each node in $\mathcal{G}$ must have at least $2f+1$ incoming edges. This is different from the case for the malicious model studied in this paper, where at least $2f$ incoming edges are required. Further discussions on the minimum requirement for our algorithm to guarantee resilient consensus are given in Section \ref{properties}.

\section{Asynchronous Network}


In this section, we analyze the MW-MSR algorithm under asynchronous updates with time delays in the communication among nodes. 

We employ the control input taking account of possible delays in the values from the neighbors as 
\begin{equation}
u_i[k]=\sum_{j\in \mathcal{N}_i^{l-}} a_{ij}[k]x_j^P[k-\tau_{ij}^P[k]],  
\end{equation}
where $x_j^P[k]$ denotes the value of node $j$ at time $k$ sent along path $P$ and $\tau_{ij}^P[k]\in \mathbb{Z}_+$ is the delay in this $(j,i)$-path $P$ at time $k$.
The delays are time varying and may be different in each path, but we assume the common upper bound $\tau$ on any \textit{normal} path $P$, over which all internal nodes are normal, as
\begin{equation}
0\leq \tau_{ij}^P[k] \leq \tau,\medspace j\in \mathcal{N}_i^{l-}, \medspace k\in \mathbb{Z}_+.
\end{equation}
Hence, each normal node $i$ becomes aware of the value
of each of its normal $l$-hop neighbor $j$ on each normal $(j,i)$-path $P$ at least once in $\tau$ time steps, but
possibly at different time instants \cite{dibaji2017resilient}. Although we impose this bound on the delays for transmission of messages, the normal nodes need neither the value of this bound nor the information that whether a path $P$ is a normal one or not.

The structure of the asynchronous MW-MSR algorithm can be outlined as follows.
At each time $k$, each normal node $i$ chooses to update or not. If it chooses not to update, then it keeps its value as $x_i[k+1]=x_i[k]$.
Otherwise, it uses the most recently received values of all its $l$-hop neighbors on each $l$-hop path to update its value using the MW-MSR algorithm in Algorithm~1. Like the one-hop MSR algorithm, if node $i$ does not receive any value along some path $P$ originating from its $l$-hop neighbor $j$ (i.e., the crash model), then node $i$ will take this value on path $P$ as an empty value and will discard this value when it applies the WM-MSR algorithm.
As we discussed earlier in Section II-E, in the asynchronous case also, manipulating message paths is equivalent to manipulating message values only and hence can be disregarded in our analysis. 

Let $D[k]$ be a diagonal matrix whose $i$th entry is
given by $d_i[k]=\sum_{j=1}^{n} a_{ij}[k].$ Then,
let the matrices $A_\gamma[k]\in \mathbb{R}^{n\times n}$ for $ 0\leq \gamma \leq \tau$ and $L_{\tau}[k]\in \mathbb{R}^{n\times (\tau +1)n}$ be given by
\begin{equation}
A_\gamma[k]=\left\{
\begin{array}{lll} 
a_{ij}[k] &\textup{if} \thinspace i\neq j \thinspace\textup{and}\thinspace \tau_{ij}[k]=\gamma,\\
0 & \textup{otherwise,}
\end{array}
\right.
\end{equation}
and $L_{\tau}[k]=\Big[ D[k]-A_0[k] \medspace -A_1[k] \medspace \cdots \medspace -A_{\tau}[k] \Big]$, respectively.
Note that the summation of each row of $L_{\tau}[k]$ is zero. 
The delay $\tau_{ij}[k]$ will be set to be one of the delays $\tau_{ij}^P[k]$ corresponding to the normal paths as we discuss further later.

Now, the control input can be expressed as
\begin{equation}
\begin{array}{lll} 
u^N[k] =-L_{\tau}^N[k]z[k],\\
u^F[k] : \textup{arbitrary,}
\end{array}
\end{equation}
where $z[k]= [x[k]^T x[k-1]^T \cdots  x[k-\tau]^T]^T$ is a $(\tau+1)n$-dimensional vector for $k\geq0$ and $L_{\tau}^N[k]$ is a matrix formed by the first $n_N$ rows of $L_{\tau}[k]$. Here, to simplify the discussion, we assume that $z[0]$ consists of the given initial values of the agents. Then, the agent dynamics can be written as
\begin{equation} \label{system2}
x[k+1]=\Gamma[k] z[k] +   \begin{bmatrix} 0  \\ I_{n_F} \end{bmatrix}u^F[k],
\end{equation}
where $\Gamma[k]$ is an  $n\times(\tau+1)n$ matrix given by $\Gamma[k] = \begin{bmatrix} I_n &  0 \end{bmatrix} -   \begin{bmatrix} L_{\tau}^N[k]^T  & 0 \end{bmatrix}^T. $

The main result of this section now follows.
Here, the safety interval differs from the synchronous case and is given by 
\begin{equation}  \label{safety2}
\mathcal{S}_{\tau}=\Big[ \min z^N[0], \max z^N[0] \Big].            	
\end{equation}

\begin{theorem}\label{asyntheorem}
	Consider a directed graph $\mathcal{G} = (\mathcal{V},\mathcal{E})$ with $l$-hop communication, where each normal node updates its value according to the asynchronous MW-MSR algorithm with parameter
	$f$. Under the $f$-total malicious model, resilient asymptotic
	consensus is achieved 
	only if the underlying graph is
	$(f + 1, f + 1)$-robust with $l$ hops. Moreover, if the underlying graph
	is $(2f + 1)$-robust with $l$ hops, then resilient consensus is attained
	with the safety interval given by \eqref{safety2}.
\end{theorem}

See the Appendix for the proof of this theorem.

\begin{remark}
	In comparison to the synchronous update case studied in the previous section, the graph condition to achieve resilient consensus under the asynchronous updates with delays is more restrictive. This is because when the nodes updates asynchronously, the normal nodes may not receive the same values from the malicious nodes, which creates a more adversarial situation for resilient consensus.
	Moreover, in the next section, in Lemma \ref{asynapplysyn}, we prove that under the $f$-total model, a graph which is $(2f + 1)$-robust with $l$ hops is also $(f + 1, f + 1)$-robust with $l$ hops. 
	For instance, the graph in Fig.~\ref{graph3} is $3$-robust with $2$ hops and hence $(2,2)$-robust with $2$ hops.
	Thus, as expected, the sufficient condition for the asynchronous algorithm guaranteeing resilient consensus is also sufficient for the synchronous algorithm.
\end{remark}

\begin{remark}
		The difference in the network requirements
		discussed above for the synchronous and asynchronous
		algorithms is a consequence partly due 
		to the deterministic nature in transmission times. 
		In \cite{dibaji2018resilient}, it is shown that for an asynchronous algorithm with one-hop
		communication without delays, we can recover the tight necessary and
		sufficient network condition of the synchronous case, i.e., the graph
		to be $(f+1,f+1)$-robust under the $f$-total malicious model. 
		It is interesting that this result requires \textit{randomization}
		in agents' communication instants whereas
		in the deterministic case, the sufficient graph condition remains
		to be $(2f+1)$-robustness, which coincides with the implication of
		Theorem \ref{asyntheorem}. In the multi-hop setting, however, delays are
		critical and hence we do not pursue such results in this paper.
\end{remark}

\begin{figure}[t]
	\centering
	\includegraphics[width=1.35in]{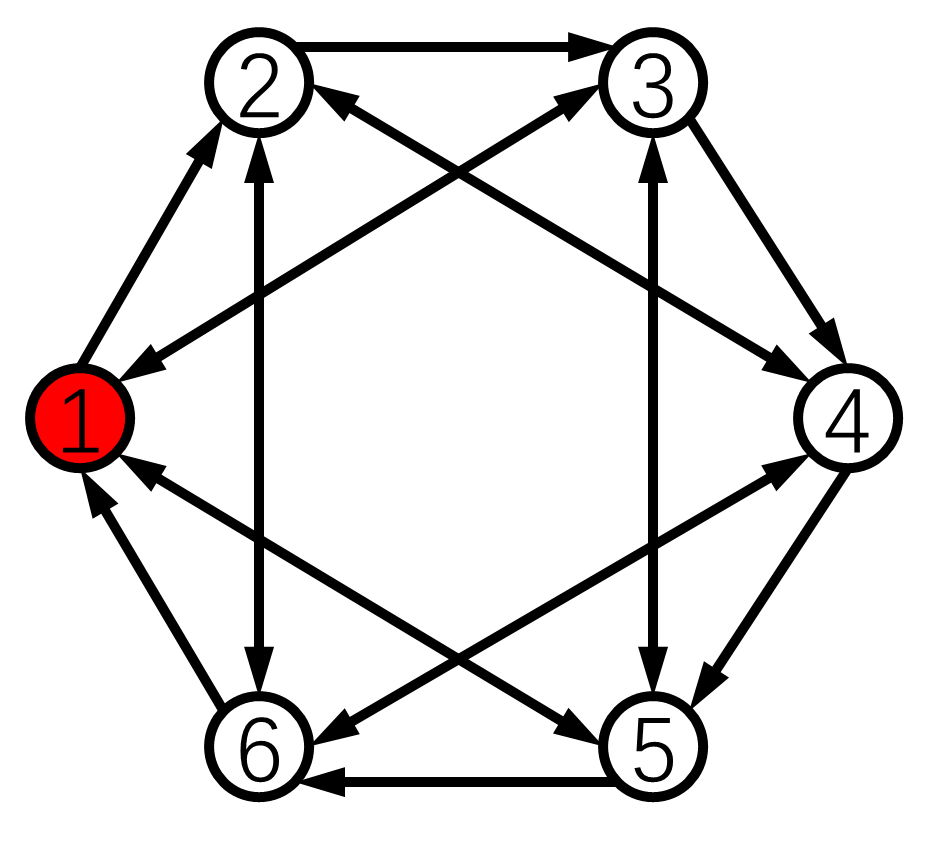}
	\vspace{-7pt}
	\caption{The graph is not $2$-robust with one hop, but it is $3$-robust with $2$ hops and $(2,2)$-robust with $2$ hops.}
	\label{graph3}
	\vspace*{-3.5mm}
\end{figure}

\section{Discussions on Graph Robustness with Multi-hop Communication}\label{properties}

In this section, we demonstrate some properties of graph robustness with $l$ hops, which generalize the properties of robustness with one hop in \cite{leblanc2013resilient} when $l=1$ for this new notion. Moreover, we provide the analysis of graph robustness with $l$ hops for the case where $l$ is sufficiently large. This corresponds to the case of unbounded path length.

\subsection{Properties of Robustness with l Hops}

As discussed earlier, our definition of robustness with $l$ hops is a generalization of the definition of robustness with one-hop communication from \cite{leblanc2013resilient}. 
Here, we are interested in investigating how properties of robustness with the one-hop case can be extended to the multi-hop case.
In what follows, we present a series of lemmas that analyze the generalized notion of robustness. 
Recall that robustness with $l$ hops is defined w.r.t. the given set $\mathcal{F}$ satisfying the $f$-total model, but we omit saying this when it is clear from the context in this section. Typically, we are interested in $f=r-1$ for $r$-robustness.

The first result is simple, stating that $(r,s)$-robustness with $l$ hops of a graph holds with smaller $r$ and $s$.

\begin{lemma}
	If a directed graph $\mathcal{G} = (\mathcal{V},\mathcal{E})$ is $(r, s)$-robust with $l$ hops, then it is also $(r', s')$-robust with $l$ hops when $0 \leq r' \leq r \ \text{and} \ 1 \leq s' \leq s$.
\end{lemma}

In the second result, we show that the level of robustness of a given graph with $l$ hops does not decrease by adding edges to the graph nor by increasing the relay range $l$.

\begin{lemma}
	Suppose that a directed subgraph $\mathcal{G} = (\mathcal{V},\mathcal{E})$ of $\mathcal{G}' = (\mathcal{V},\mathcal{E}')$ is $(r, s)$-robust with $l$ hops, where $\mathcal{E} \subseteq\mathcal{E}'$. Then $\mathcal{G}'$ is $(r, s)$-robust with $l$ hops.
	Moreover, $\mathcal{G}$ is $(r, s)$-robust with $l'$ hops, where $l'\geq l$.
\end{lemma}

The maximum robustness with $l$ hops for a graph consisting of $n$ nodes is the same as that with the one-hop case. The following lemma suggests that the bound $n \geq 2f+1$ cannot be breached by introducing multi-hop communication to the MSR algorithms. This bound is dependent on the nature of the MSR algorithms, which cannot tolerate half or more of the agents to be adversarial.

\begin{lemma}\label{maxrobust}
	No directed graph  $\mathcal{G} = (\mathcal{V},\mathcal{E})$ on $n$ nodes is $(
	\lceil n/2\rceil + 1)$-robust with $l$ hops. Moreover, the complete graph $\mathcal{K}_n$ is $(
	\lceil n/2\rceil, s)$-robust with $l$ hops for $1 \leq s \leq n$.
\end{lemma}

\begin{proof}
	We consider the nontrivial case with $n\geq 3$. Then pick $\mathcal{V}_1$ and
	$\mathcal{V}_2$ by taking any bipartition of $\mathcal{V}$ (i.e., $\mathcal{V}_1\cap \mathcal{V}_2 =\emptyset$ and  $\mathcal{V}_1 \cup \mathcal{V}_2 =\mathcal{V}$)
	such that $|\mathcal{V}_1| = \lceil n/2\rceil$ and $|\mathcal{V}_2| = \lfloor n/2\rfloor$. Neither $\mathcal{V}_1$ nor $\mathcal{V}_2$ has
	$\lceil n/2\rceil+1$ nodes; thus, neither one has a node being $(\lceil n/2\rceil+1)$-reachable with $l$-hop communication. Hence, $\mathcal{G}$ is not $(\lceil n/2\rceil+1)$-robust with $l$ hops.
	
	The proof for the complete graph case follows an analysis similar to the one for the one-hop case \cite{leblanc2013resilient}.
\end{proof}

The following lemma exposes that the notion of $(r, s)$-robust graphs has a more complicated structure in the relation between the two parameters $r$ and $s$.

\begin{lemma}\label{2fplus1}
		If a directed graph $\mathcal{G} = (\mathcal{V},\mathcal{E})$ is $(r, s)$-robust with $l$ hops under the $f$-total model, then it is also $(r-1, s+1)$-robust with $l$ hops under the $f$-total model.
\end{lemma}

\begin{proof}
	Consider any nonempty disjoint subsets $\mathcal{V}_1, \mathcal{V}_2\subset \mathcal{V}$. If $|\mathcal{Z}_{\mathcal{V}_a}^r| = |\mathcal{V}_a|$ or $|\mathcal{Z}_{\mathcal{V}_a}^{r-1}| = |\mathcal{V}_a|$ for $a=1, 2$, then this pair of subsets satisfies conditions 1) or 2) of $(r-1,s+1)$-robustness with $l$ hops. Thus, assume $|\mathcal{Z}_{\mathcal{V}_a}^r| < |\mathcal{V}_a|$ and $|\mathcal{Z}_{\mathcal{V}_a}^{r-1}| < |\mathcal{V}_a|$ for $a=1, 2$. Then condition 3) for $(r,s)$-robustness with $l$ hops must be satisfied, i.e., $| \mathcal{Z}_{\mathcal{V}_1}^r| +| \mathcal{Z}_{\mathcal{V}_2}^r| \geq s$.
	
	Since $s\geq1$, at least one of $\mathcal{Z}_{\mathcal{V}_1}^r$, $\mathcal{Z}_{\mathcal{V}_2}^r$ is nonempty. Suppose that $\mathcal{Z}_{\mathcal{V}_1}^r$ is nonempty. Then, we choose $i\in \mathcal{Z}_{\mathcal{V}_1}^r$. Pick the new pair of nonempty disjoint subsets as $\mathcal{V}_1'= \mathcal{V}_1\setminus \{i\}$ and $\mathcal{V}_2'= \mathcal{V}_2$.
	Observe that if $j \in  \mathcal{Z}_{\mathcal{V}_1'}^r$ then $j \in  \mathcal{Z}_{\mathcal{V}_1}^{r-1}$.
	Because node $i$ is the only difference between $\mathcal{V}_1'$ and $ \mathcal{V}_1$, and node $i$ can only be part of one independent path whose destination is node $j$.
	Consequently, even if node $i$ is on one independent path whose destination is node $j$, node $j$ still has the $(r-1)$-reachable property, i.e., it has at least $r-1$ other independent paths of at most $l$ hops originating from nodes outside $\mathcal{V}_1$ and these paths do not have any internal nodes from set $\mathcal{F}$. Then, by adding $i\in \mathcal{Z}_{\mathcal{V}_1}^r\subseteq \mathcal{Z}_{\mathcal{V}_1}^{r-1}$ back into the set $\mathcal{V}_1$, we have
	\begin{equation}  \label{ineq1}
	|\mathcal{Z}_{\mathcal{V}_1}^{r-1}| \geq |\mathcal{Z}_{\mathcal{V}_1'}^r| +1	.
	\end{equation}
	Moreover, $|\mathcal{Z}_{\mathcal{V}_1'}^r| < |\mathcal{V}_1'|$ and $|\mathcal{Z}_{\mathcal{V}_2'}^r| < |\mathcal{V}_2'|$ (since $\mathcal{V}_2'=\mathcal{V}_2$). The first inequality holds because if $|\mathcal{Z}_{\mathcal{V}_1'}^r| = |\mathcal{V}_1'|$ and from the observation we just proved (if $j \in  \mathcal{Z}_{\mathcal{V}_1'}^r$ then $j \in  \mathcal{Z}_{\mathcal{V}_1}^{r-1}$), we have
	$|\mathcal{Z}_{\mathcal{V}_1}^{r-1}| =|\mathcal{V}_1|$, leading us to a contradiction.
	
	Therefore, for nonempty disjoint subsets $\mathcal{V}_1', \mathcal{V}_2'$, it must hold that $| \mathcal{Z}_{\mathcal{V}_1'}^r| +| \mathcal{Z}_{\mathcal{V}_2'}^r| \geq s$. 
	Together with \eqref{ineq1}, we have 
	\begin{equation*}  
	| \mathcal{Z}_{\mathcal{V}_1}^{r-1}| +| \mathcal{Z}_{\mathcal{V}_2}^{r-1}| \geq | \mathcal{Z}_{\mathcal{V}_1'}^r| +| \mathcal{Z}_{\mathcal{V}_2'}^r|+1\geq s+1	,
	\end{equation*} 
	suggesting that $\mathcal{G}$ is $(r-1, s+1)$-robust with $l$ hops under the $f$-total model.
\end{proof}

From this result, we can directly derive the following corollary, which indicates that if a graph is $(2f+1)$-robust with $l$ hops, then it is also $(f+1,f+1)$-robust with $l$ hops.

\begin{corollary}\label{asynapplysyn}
	If a directed graph $\mathcal{G} = (\mathcal{V},\mathcal{E})$ is $(r+s-1)$-robust with $l$ hops, where $1\leq r+s-1\leq \lceil n/2\rceil$, then $\mathcal{G}$ is $(r,s)$-robust with $l$ hops.
\end{corollary}

Similar to the one-hop case, robustness with $l$ hops can guarantee a certain level of graph connectivity and a minimum in-degree of the graph. The proof for the connectivity result is trivial and hence omitted.

\begin{lemma}
	If a directed graph $\mathcal{G} = (\mathcal{V},\mathcal{E})$ is $r$-robust with $l$ hops, then $\mathcal{G}$ is at least $r$-connected.
\end{lemma}

\begin{lemma}\label{mini-degree}
	If a directed graph $\mathcal{G} = (\mathcal{V},\mathcal{E})$ is $(r,s)$-robust with $l$ hops under the $(r-1)$-total model, where $0\leq r\leq \lceil n/2\rceil$ and $1\leq s\leq n$, then $\mathcal{G}$ has the minimum in-degree $\delta(\mathcal{G})$ as
	\begin{equation*}
	\delta(\mathcal{G}) \geq \left\{
	\begin{array}{lll} 
	r+s-1 &\textup{if} \ s<r,\\
	2r-2 & \textup{if}\ s\geq r.
	\end{array}
	\right.
	\end{equation*}
\end{lemma}
\vspace{0.12cm}

\begin{proof}
	The case for $n\leq 2$ and $r\leq 1$ is trivial. Consider the case for $n\geq 3$ and $2\leq r\leq \lceil n/2\rceil$.
	Fix node $i\in \mathcal{V}$. Choose $\mathcal{V}_1=\{i\}$ and $\mathcal{V}_2= \mathcal{V}\setminus \mathcal{V}_1$. Then, $| \mathcal{Z}_{\mathcal{V}_2}^r|=0$ and $| \mathcal{Z}_{\mathcal{V}_1}^r| =|\mathcal{V}_1|$. Hence, node $i$ must have at least $r$ independent paths from outside and the number of the in-neighbors of node $i$ should be $|\mathcal{N}_i^-|\geq r$. (Note that the malicious nodes can be the source nodes of these independent paths.)
	
	When $s<r$, form $\mathcal{V}_1$ by choosing $s-1$ in-neighbors of node $i$ along with node $i$ itself. Then, choose $\mathcal{V}_2= \mathcal{V}\setminus \mathcal{V}_1$. 
	Since $|\mathcal{V}_1| = s < r$, then, $| \mathcal{Z}_{\mathcal{V}_2}^r|=0$ and $| \mathcal{Z}_{\mathcal{V}_1}^r| =|\mathcal{V}_1|$. The worst case is that the $s-1$ in-neighbors of node $i$ are all malicious, node $i$ should have at least $r$ independent paths from outside, and these paths do not contain any malicious nodes as internal nodes.
	This implies that node $i$ has additional
	$r$ in-neighbors outside of $\mathcal{V}_1$, thereby guaranteeing $|\mathcal{N}_i^-|\geq r+s-1$. 
	
	When $s \geq r$, form $\mathcal{V}_1$ by choosing $r-2$ in-neighbors of node $i$ along with node $i$ itself. Then, choose
	$\mathcal{V}_2= \mathcal{V}\setminus \mathcal{V}_1$. Since $|\mathcal{V}_1| < r$ and $s \geq r$, we have $| \mathcal{Z}_{\mathcal{V}_2}^r|=0$ and $| \mathcal{Z}_{\mathcal{V}_1}^r| =|\mathcal{V}_1|$. 
	Consider the worst case that the $r-2$ in-neighbors of node $i$ are all malicious.
	It must be that node $i$ has additional $r$ in-neighbors outside of $\mathcal{V}_1$, thereby guaranteeing $|\mathcal{N}_i^-|\geq 2r-2$. Since we choose $i\in \mathcal{V}$
	arbitrarily, we have proved the bound for $\delta(\mathcal{G})$.
\end{proof}

Here, we have extended the proof in \cite{leblanc2013resilient} to the multi-hop case. 
From Lemma \ref{mini-degree}, we conclude that 
$\mathcal{G}$ should have the minimum in-degree no less than $2f$
to guarantee resilient consensus using the MW-MSR algorithm under the $f$-total malicious model. This holds since the underlying graph $\mathcal{G}$ is at least $(f+1,f+1)$-robust with $l$ hops for achieving resilient consensus. 
For MSR algorithms, nodes with $2f$ in-neighbors may not use any values from neighbors, which may appear problematic especially if there are multiple such nodes. However, there are examples such as cycle graphs\footnote[3]{A cycle graph is an undirected graph consisting of only a single cycle.} satisfying $(f+1,f+1)$-robust with $l$ hops, while every node in cycle graphs has in-degree as $2f$. Detailed analysis is given in the next subsection.

\subsection{The Case of Unbounded Path Length}
In this subsection, we discuss the relation of the graph conditions used in this paper and in the recent work \cite{khan2020exact}. The authors there studied the Byzantine binary consensus under the local broadcast model, which is essentially equivalent to the $f$-total malicious mode studied in the current paper. The proposed algorithm in \cite{khan2020exact} is based on a non-iterative flooding algorithm, where nodes must relay their values over the entire network along with the path information. This model corresponds to the case of unbounded path length in our work, i.e. $l\geq l^*$, where $l^*$ is the longest cycle-free path length of the network.
Moreover, they propose a tight necessary and sufficient graph condition for their algorithm to achieve binary consensus under synchronous updates. 
Our aim in this part of the paper is to establish that our graph condition is equivalent to theirs for the case of unbounded path length ($l\geq l^*$). Further, we will highlight that to achieve the same tolerance as the algorithm in \cite{khan2020exact}, our algorithm does not in general require $l^*$-hop communication necessarily for general graphs.

To show the equivalence between the two graph conditions, we introduce some graph notions from \cite{khan2020exact}. 
There, normal nodes update the states based on a modified certified propagation algorithm \cite{tseng2015broadcast}, i.e., when a normal node receives $f+1$ same binary values from different paths excluding a suspicious set $\mathcal{F}$, it commits its value to this value. Hence, their graph notion is closely related to the partitions of sets $\mathcal{V}$ and $\mathcal{F}$.

\begin{definition}
	For disjoint node sets $\mathcal{X}, \mathcal{Y}$, we say $\mathcal{X}\rightarrow \mathcal{Y}$ if and only if set $\mathcal{X}$ contains at least $f+1$ distinct incoming neighbors of $\mathcal{Y}$, i.e., $|\{i:(i, j)\in \mathcal{E}, i \in \mathcal{X}, j \in \mathcal{Y}\}| >f$. Denote $\mathcal{X}\not\rightarrow \mathcal{Y}$ when $\mathcal{X}\rightarrow \mathcal{Y}$ is not true.
	
\end{definition}

\begin{definition}
	For disjoint node sets $\mathcal{X}, \mathcal{Y}$ and for set $\mathcal{F}$, we say $\mathcal{X}\stackrel{\mathcal{F}}{\rightsquigarrow} \mathcal{Y}$ if and only if for every node $u\in \mathcal{Y}$, there exist at least $f+1$ disjoint
	$\mathcal{X}u$-paths that have only $u$ in common and none of them contains any internal node from the set $\mathcal{F}$. Denote $\mathcal{X}\stackrel{\mathcal{F}}{\not\rightsquigarrow} \mathcal{Y}$ when $\mathcal{X}\stackrel{\mathcal{F}}{\rightsquigarrow} \mathcal{Y}$ is not true.
\end{definition}

Two graph notions are introduced next, called conditions NC and SC. They are known to be equivalent \cite{khan2020exact}.

\begin{definition}
	\textit{(Condition NC)}
	Given a graph $\mathcal{G} = (\mathcal{V},\mathcal{E})$, condition NC is said to hold if for every
	partition $\mathcal{L},\mathcal{C},\mathcal{R}$ of $\mathcal{V}$, and for every set $\mathcal{F}$ with $|\mathcal{F}|\leq f$,
	where both $\mathcal{L} \setminus \mathcal{F}$ and $\mathcal{R} \setminus \mathcal{F}$ are non-empty, we have that either $\mathcal{R} \cup \mathcal{C}\rightarrow \mathcal{L} \setminus \mathcal{F}$ or $\mathcal{L} \cup \mathcal{C} \rightarrow \mathcal{R} \setminus \mathcal{F}$.
	
	\textit{(Condition SC)}
	Given a graph $\mathcal{G} = (\mathcal{V},\mathcal{E})$, condition SC is said to hold if for every
	partition $\mathcal{L},\mathcal{R}$ of $\mathcal{V}$, and for every set $\mathcal{F}$ with $|\mathcal{F}|\leq f$,
	where both $\mathcal{L} \setminus \mathcal{F}$ and $\mathcal{R} \setminus \mathcal{F}$ are non-empty, we have that either $\mathcal{L}\stackrel{\mathcal{F}}{\rightsquigarrow} \mathcal{R}\setminus\mathcal{F}$ or $\mathcal{R}\stackrel{\mathcal{F}}{\rightsquigarrow} \mathcal{L}\setminus\mathcal{F}$.
\end{definition}

We are now ready to show that condition NC (and hence condition SC) is equivalent to our robust graph notion with multi-hop communication used in Theorem \ref{syn}.

\begin{proposition}\label{connectivity}
	Consider a directed graph $\mathcal{G} = (\mathcal{V},\mathcal{E})$ with $l$-hop communication where $l\geq l^*$.
	The graph $\mathcal{G}$ is $(f + 1, f + 1)$-robust with $l$ hops if and only if condition NC holds.
\end{proposition}

\begin{proof}
	We first show the only if part. Since the graph is $(f+1,f+1)$-robust with $l$ hops under the $f$-total model, at least one of the three conditions in Definition \ref{rs-robust} holds. For every partition $\mathcal{L},\mathcal{R}$ of $\mathcal{V}$, and for every set $|\mathcal{F}|\leq f$, where both $\mathcal{L} \setminus \mathcal{F}$ and $\mathcal{R} \setminus \mathcal{F}$ are non-empty, we conclude that there is at least one normal node $i$ in the union of $\mathcal{L}$ and $\mathcal{R}$ that has $f+1$ independent paths from outside and these paths do not contain any internal nodes in $\mathcal{F}$. This can be seen in the proof of Theorem \ref{syn}. Suppose that node $i\in \mathcal{L}$ has this property. 
	For each independent path, there exists at least one edge that goes from the outside of $\mathcal{L}$ to a node in $\mathcal{L}$, and thus, $\mathcal{R} \cup \mathcal{C}\rightarrow \mathcal{L} \setminus \mathcal{F}$. The case for $i\in \mathcal{R}$ can be proved similarly.

	Next, we show the if part by contradiction. Suppose that $\mathcal{G}$ is not $(f + 1, f + 1)$-robust with $l$ hops. Then, none of the three conditions in Definition \ref{rs-robust} holds. For every partition $\mathcal{L},\mathcal{R}$ of $\mathcal{V}$, and for every set $|\mathcal{F}|\leq f$, where both $\mathcal{L} \setminus \mathcal{F}$ and $\mathcal{R} \setminus \mathcal{F}$ are non-empty, we conclude that all the normal nodes in the union of $\mathcal{L}$ and $\mathcal{R}$ have at most $f$ independent paths from outside where these paths do not contain any internal nodes in $\mathcal{F}$. Hence, $\mathcal{L}\stackrel{\mathcal{F}}{\not\rightsquigarrow} \mathcal{R}\setminus\mathcal{F}$ and $\mathcal{R}\stackrel{\mathcal{F}}{\not\rightsquigarrow} \mathcal{L}\setminus\mathcal{F}$ (i.e., condition SC does not hold), and thus we have contradiction. 
	
	Finally, since condition SC and condition NC are equivalent, we have proved that condition NC implies the $(f + 1, f + 1)$-robustness with $l$ hops.
\end{proof}

Although our condition coincides with those in \cite{khan2020exact} when $l\geq l^*$, we note that the maximum robustness of a given graph does not require $l\geq l^*$ necessarily. That is, for a given graph under the $f$-total model, our algorithm may not require $l^*$-hop communication to get the same tolerance as the algorithm using $l^*$-hop communication in \cite{khan2020exact}. We illustrate this fact by presenting the examples of cycle graphs.

\begin{figure}[t]
	\centering
	\includegraphics[width=1.35in]{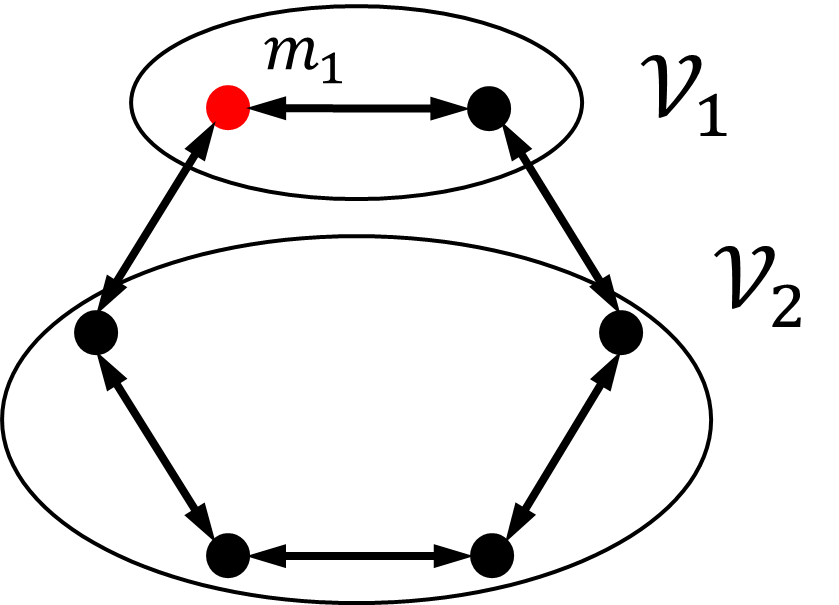}
	\vspace{-7pt}
	\caption{Illustration for cycle graphs.}
	\label{cyclegraph}
	\vspace*{-3.5mm}
\end{figure}

From the example graph in Fig.~\ref{graph1}(a), we can see that 
for any cycle graph, the following lemma holds. This indicates that resilient consensus is guaranteed using the MW-MSR algorithm in such graphs when one node behaves maliciously. In \cite{khan2020exact}, it is reported that a cycle graph can tolerate one malicious node, but their algorithm uses the $l^*$-hop communication. In this respect, as the following lemma suggests, our algorithm is efficient by exploiting the ability of the MW-MSR algorithm and the graph condition is tighter than the one in \cite{khan2020exact}. Moreover, note that a cycle graph is 2-connected, which is the minimum connectivity requirement for any MSR-based algorithm to guarantee resilient consensus for $f=1$. 

\begin{lemma}
	The cycle graph $\mathcal{C}_n$ with $n>2$ nodes is $(2,2)$-robust with $\lceil l^*/2 \rceil$ hops under the $1$-total model.
\end{lemma}

\begin{proof}
	We need to show that for any node partition $\mathcal{V}_1$,
	$\mathcal{V}_2$ of $\mathcal{V}$, at least one of the conditions for $(2,2)$-robustness with $l$ hops holds. Let $\mathcal{F}$ be a set of a single node, i.e., $\mathcal{F}=\{m_1\}$ (satisfying the 1-total model).
	We first select a single node as set $\mathcal{V}_1$. Then, for any set $\mathcal{F}$ and for any set $\mathcal{V}_2$, condition 1) for $(2,2)$-robustness with $l$ hops holds. (The case is similar when we select non-neighboring nodes as set $\mathcal{V}_1$.)
	Second, we select two neighboring nodes as set $\mathcal{V}_1$. If node $m_1\notin \mathcal{V}_1$, then condition 1) for $(2,2)$-robustness with 2 hops holds. If node $m_1\in \mathcal{V}_1$, then node $m_1$ has 2 independent 2-hop paths originating from outside. To meet the conditions for $(2,2)$-robustness with $l$ hops, we need to find another node having this property in $\mathcal{V}_2$
	(see the illustration in Fig. \ref{cyclegraph}).
	The worst case is when all the remaining nodes are in $\mathcal{V}_2$. Then the middle node in $\mathcal{V}_2$ has 2 shortest paths originating from outside, which are of length $\lceil l^*/2 \rceil$ hops.
	
	We can continue this process and select three neighboring nodes as set $\mathcal{V}_1$. We can follow an analysis as above: If node $m_1\in \mathcal{V}_1$ and all the remaining nodes are in $\mathcal{V}_2$, then the middle node in $\mathcal{V}_2$ has 2 shortest paths originating from outside, which are shorter than the length $\lceil l^*/2 \rceil$ hops. This process can be continued until we switch sets $\mathcal{V}_1$ and $\mathcal{V}_2$. Hence, we conclude that the cycle graph $\mathcal{C}_n$ is $(2,2)$-robust with $\lceil l^*/2 \rceil$ hops.
\end{proof}

\section{Numerical Examples}

In this section, we conduct numerical simulations over networks
using both synchronous and asynchronous versions of the proposed MW-MSR algorithm to verify their effectiveness. 

\begin{figure}[t]
	\centering
	\includegraphics[width=3.2in,height=1.3in]{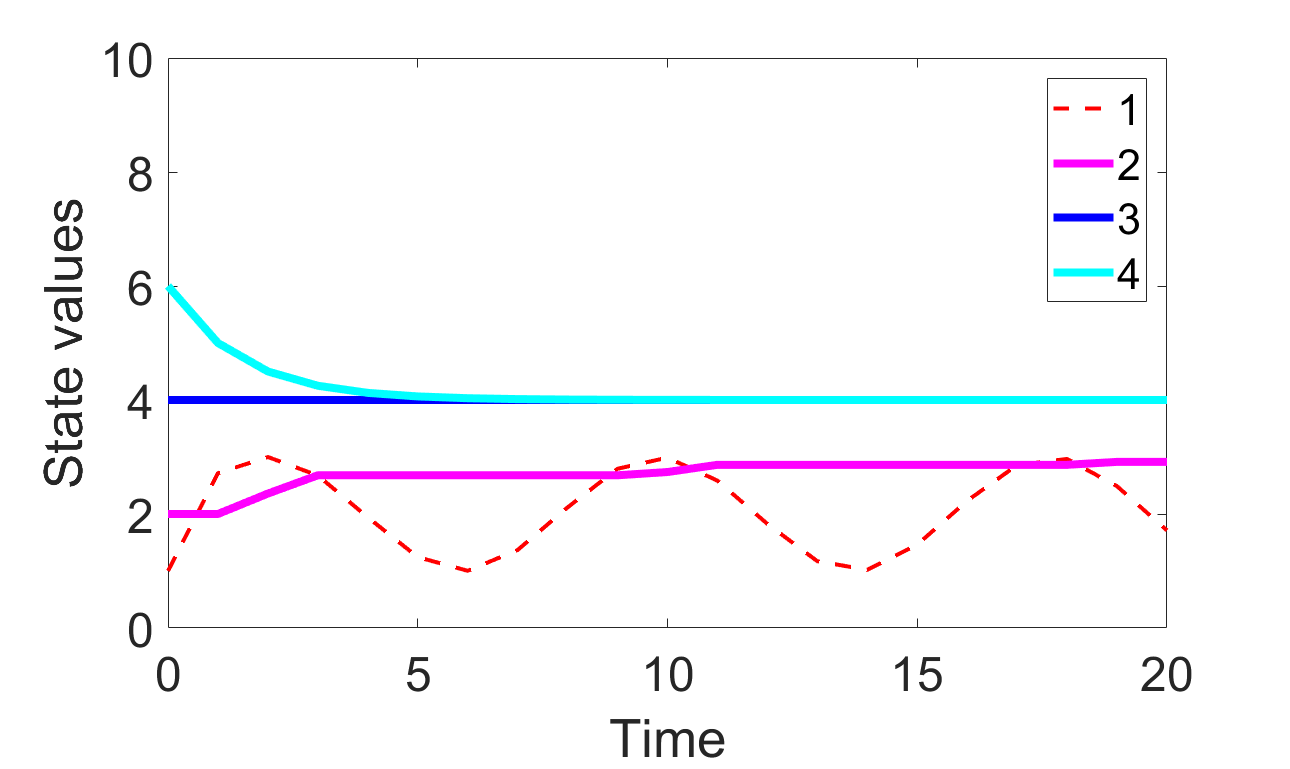}
	\vspace{-7pt}
	\caption{Time responses of the synchronous one-hop W-MSR algorithm.}
	\label{one-hop-state}
	\vspace*{-3.5mm}
\end{figure}

\begin{figure}[t]
	\centering
	\includegraphics[width=3.2in,height=1.3in]{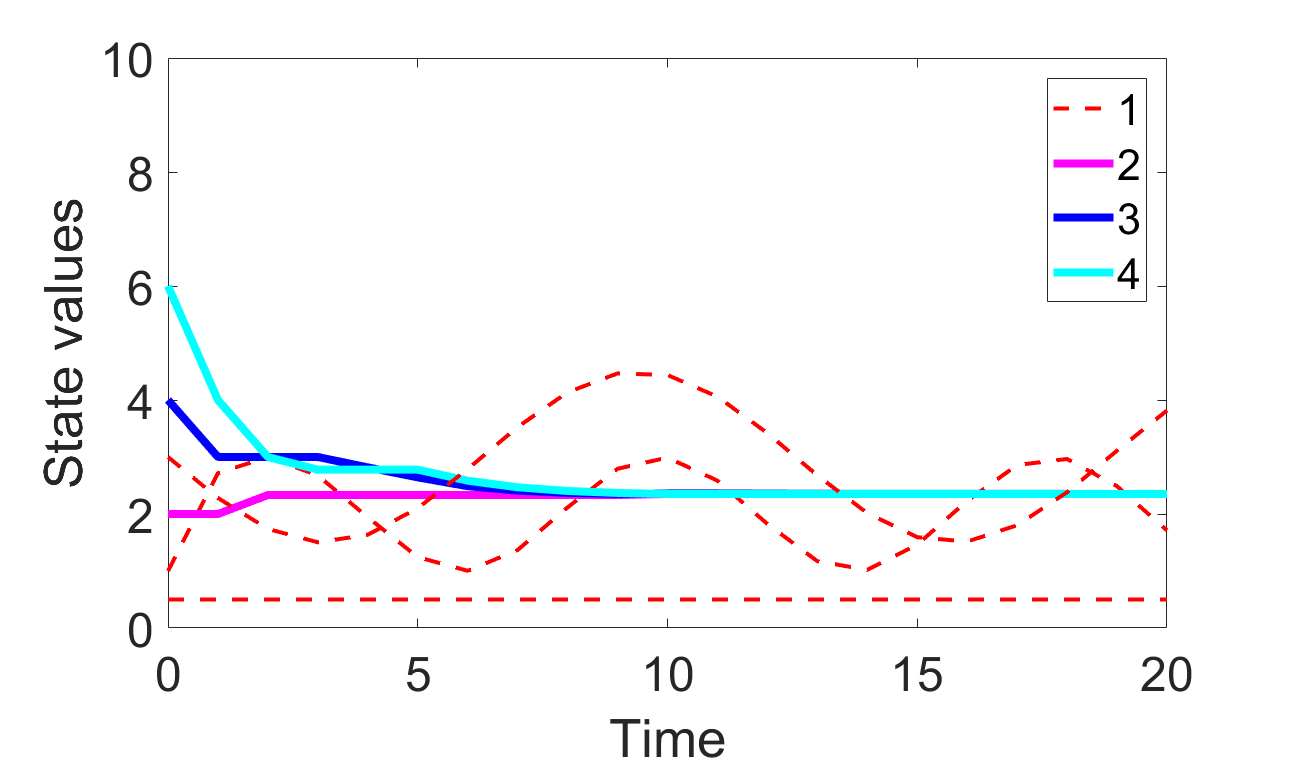}
	\vspace{-7pt}
	\caption{Time responses of the synchronous two-hop MW-MSR algorithm.}
	\label{two-hop-syn}
	\vspace*{-3.5mm}
\end{figure}

\begin{figure}[t]
	\centering
	\includegraphics[width=3.2in,height=1.3in]{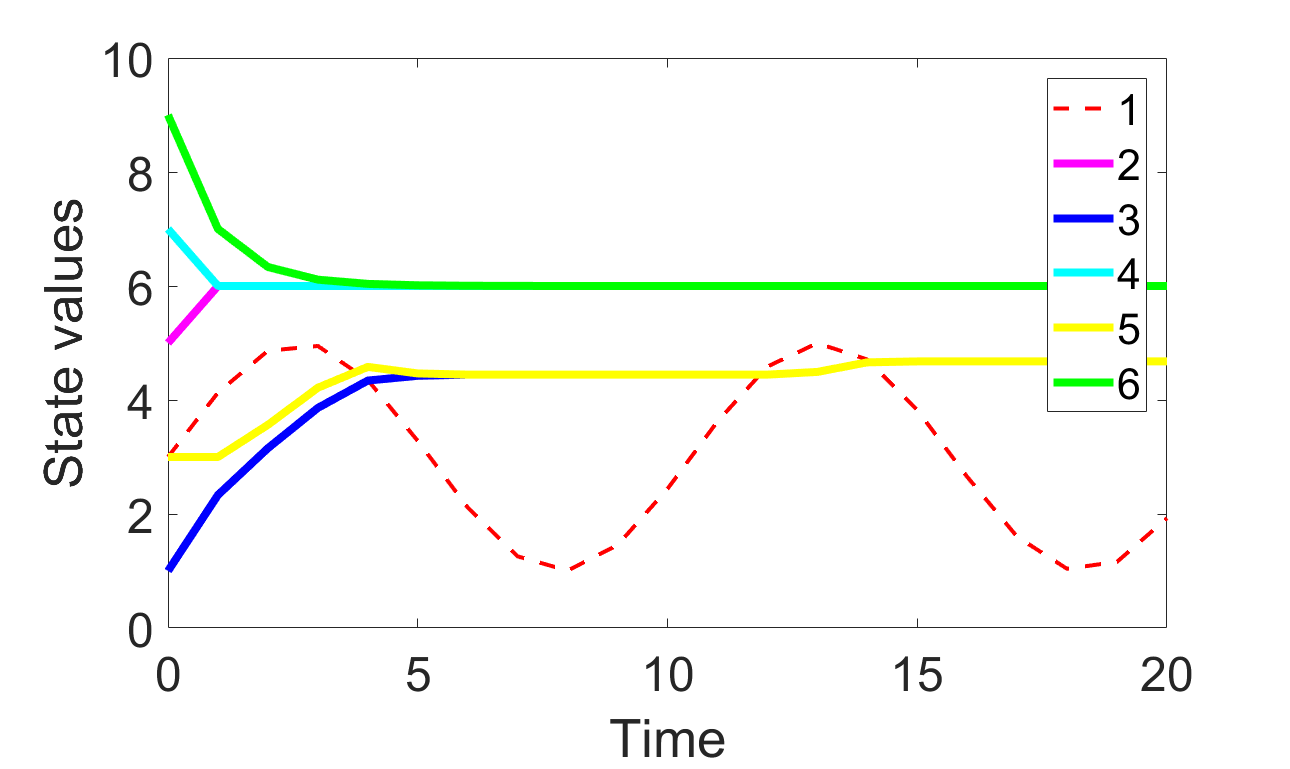}
	\vspace{-7pt}
	\caption{Time responses of the synchronous one-hop W-MSR algorithm.}
	\label{6_one-hop-state}
	\vspace*{-3.5mm}
\end{figure}

\begin{figure}[t]
	\centering
	
	\subfigure[\scriptsize{The states of all nodes.}]{
		\includegraphics[width=3.2in,height=1.3in]{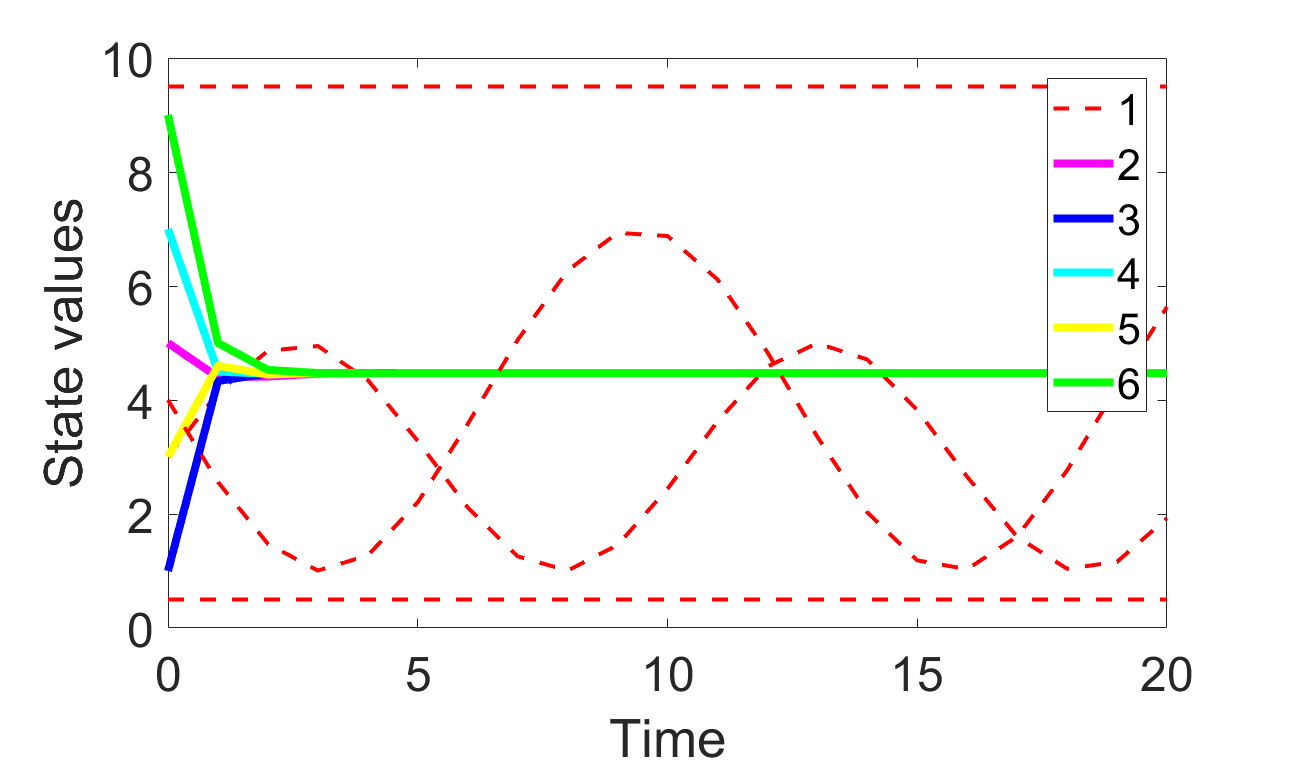}
	}
	
	\vspace{-7pt}
	\subfigure[\scriptsize{Consensus error.}]{
		\includegraphics[width=3.2in,height=1.3in]{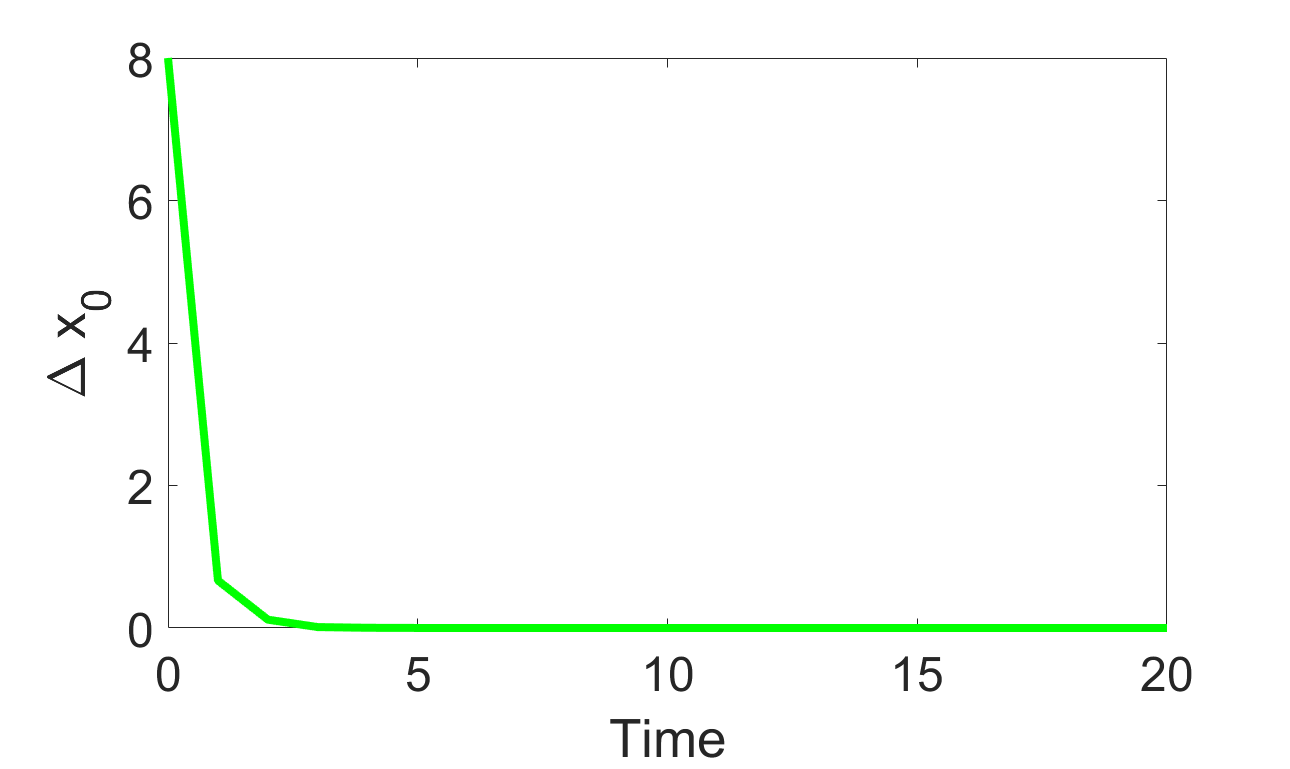}
	}
	\vspace{-7pt}
	\caption{Time responses of the synchronous two-hop MW-MSR algorithm.}
	\label{6_two-hop-syn}
	\vspace*{-3.5mm}
\end{figure}

\subsection{Synchronous MW-MSR Algorithm}

In this part, we conduct simulations for the synchronous MW-MSR algorithm.
Consider the undirected network in Fig.~\ref{graph1}(a) with $f=1$. 
Let the initial states be $x[0]=[1\ 2\ 4\ 6]^T$. 
This graph is not $(2, 2)$-robust with one hop, and hence, is not robust enough to tolerate $f=1$ using the conventional one-hop W-MSR algorithm. 
Here we set node 1 to be malicious and let the value of node 1 evolve based on the sine function w.r.t. time.
Then, normal nodes update their values using the one-hop W-MSR algorithm. 
The results are given in Fig.~\ref{one-hop-state}, and resilient consensus among normal nodes is not achieved.

Next, consider the two-hop case for this graph. For malicious node 1, we assume that it does not only manipulate its own value as in the one-hop case, but also relays false information. Specifically, when node 1 receives a message from node 4 and relays the value $x_4[k]$ to node 2, it manipulates this value based on the sine function w.r.t. time. Similarly, when node 1 receives a message from node 2 and relays the value $x_2[k]$ to node 4, it manipulates this value to a fixed value of 0.5.
Then, we observe that resilient consensus is achieved as shown in Fig.~\ref{two-hop-syn}.

\subsection{Asynchronous MW-MSR Algorithm}

In this part, we conduct simulations for the asynchronous MW-MSR algorithm.
Consider the directed network in Fig.~\ref{graph3} with $f=1$. 
Let the nodes take initial states as $x[0]=[3\ 5\ 1\ 7\ 3\ 9]^T$. 
This graph is not $2$-robust with one hop (e.g., consider the sets $\{1,3,5\}$ and $\{2,4,6\}$), and hence, is not robust enough to tolerate $f=1$ using the one-hop W-MSR algorithm. 
Here, we set node 1 to be malicious and assume that the value of node 1 evolves based on the sine function w.r.t. time.
Then, we apply the one-hop W-MSR algorithm and observe that resilient consensus among normal nodes is not achieved as shown in Fig.~\ref{6_one-hop-state}.

Next, consider the two-hop case for this graph. It becomes $3$-robust with $2$ hops, and hence, it is also $(2,2)$-robust with $2$ hops as Corollary \ref{asynapplysyn} indicates.
Therefore, in the two-hop case, it can tolerate one malicious node under both synchronous and asynchronous updates.
For malicious node 1, we assume that it does not only manipulate its own value as in the one-hop case, but also relays false information. Specifically, when node 1 receives a message from node 3 and relays the value $x_3[k]$ to its other neighbors, it manipulates this value to a fixed value of 0.5. Similarly, when node 1 receives a message from node 6 and relays the value $x_6[k]$ to other neighbors, it manipulates this value to a fixed value of 9.5. Additionally, when node 1 receives a message from node 5 and relays the value $x_5[k]$ to other neighbors, it manipulates this value based on the sine function w.r.t. time.
In Fig.~\ref{6_two-hop-syn}, we plot the consensus error given by $\Delta x_0[k]= \max x^N[k]- \min x^N[k]$ and observe that resilient consensus is achieved.

Lastly, we examine the two-hop algorithm under asynchronous updates with delays. We consider the same attack, but let each normal node update in a periodic manner. Specifically, nodes 2, 3, 4, 5, and 6 update in every 1, 5, 4, 3, and 2 steps, respectively. The delays for the messages from one-hop neighbors and two-hop neighbors are set as 0 and 1 step, respectively.
Fig.~\ref{6_two-hop-asyn} shows the states as well as the consensus error given by $\Delta x_\tau[k]= \max z^N[k]- \min z^N[k]$. It indicates that consensus is attained despite the malicious attacks.
Note that in this situation, we can only guarantee the nonincreasing property of $\Delta x_\tau[k]$ (with $\tau=5$).
Through these simulations, we have verified the effectiveness of the MW-MSR algorithms to achieve resilient consensus in small-scale networks.

\begin{figure}[t]
	\centering
	
	\subfigure[\scriptsize{The states of all nodes.}]{
		\includegraphics[width=3.2in,height=1.3in]{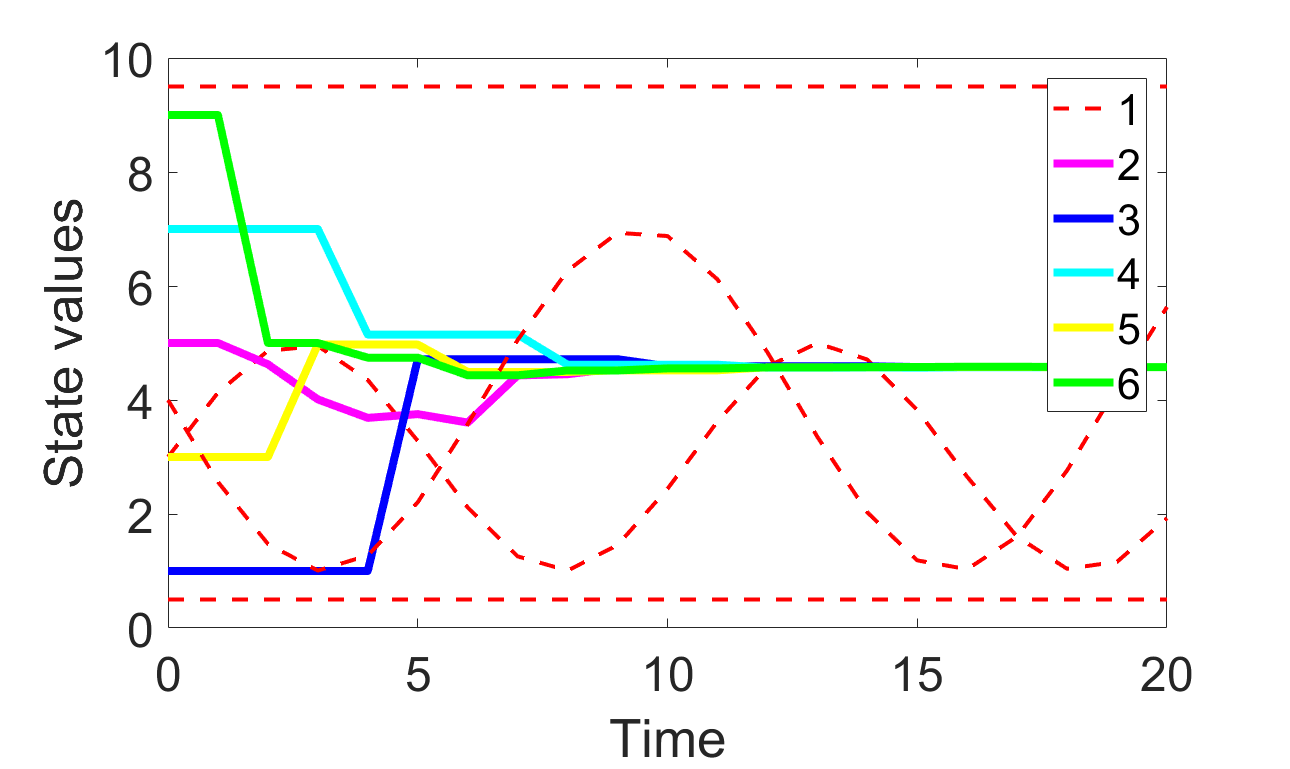}
	}
	
	\vspace{-7pt}
	\subfigure[\scriptsize{Consensus error.}]{
		\includegraphics[width=3.2in,height=1.3in]{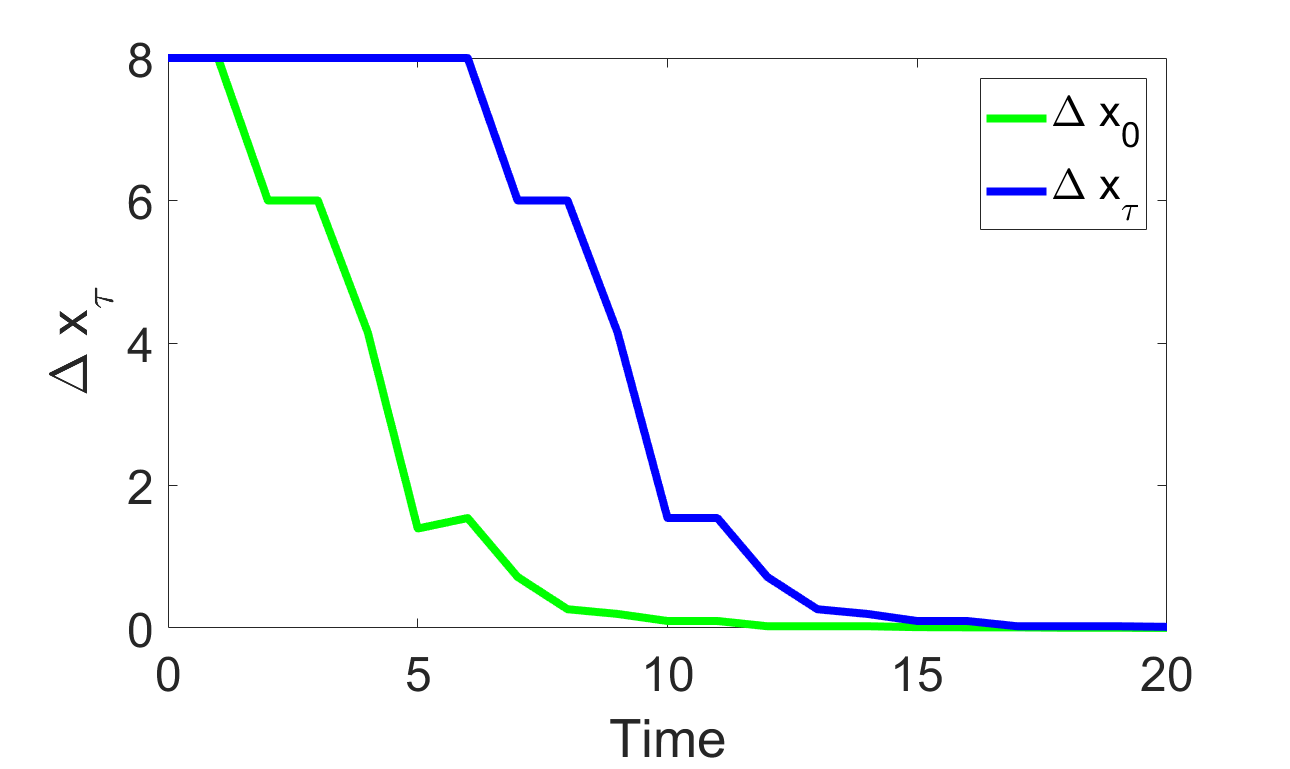}
	}
	\vspace{-7pt}
	\caption{Time responses of the asynchronous two-hop MW-MSR algorithm.}
	\label{6_two-hop-asyn}
	\vspace*{-3.5mm}
\end{figure}

\subsection{Simulations in Large Wireless Sensor Networks}

In this part of the simulations, we create a WSN composed of 100 nodes located in a grid structure as shown in Fig.~\ref{location}. Let the nodes take indices $0, 1,\dots, 99$ and the coordinate of node $i$ is $( i\mod 10 , \lfloor \frac{i}{10} \rfloor)$. 
Each node can communicate only with the nodes located within the communication radius of $r$. 
Once $r$ is determined, the topology of the network is formed. 
Then we apply the one-hop, two-hop, and three-hop MW-MSR algorithms to the network.
Recall that $f$ denotes the maximum number of malicious nodes in the network, and we increase $f$ from 0 to 11 by selecting the malicious nodes with indices in the order of $32, 34, 36, 38, 43, 62, 64, 66, 68, 74, 14$.

\begin{figure}[]
	\centering
	\includegraphics[width=1.8in]{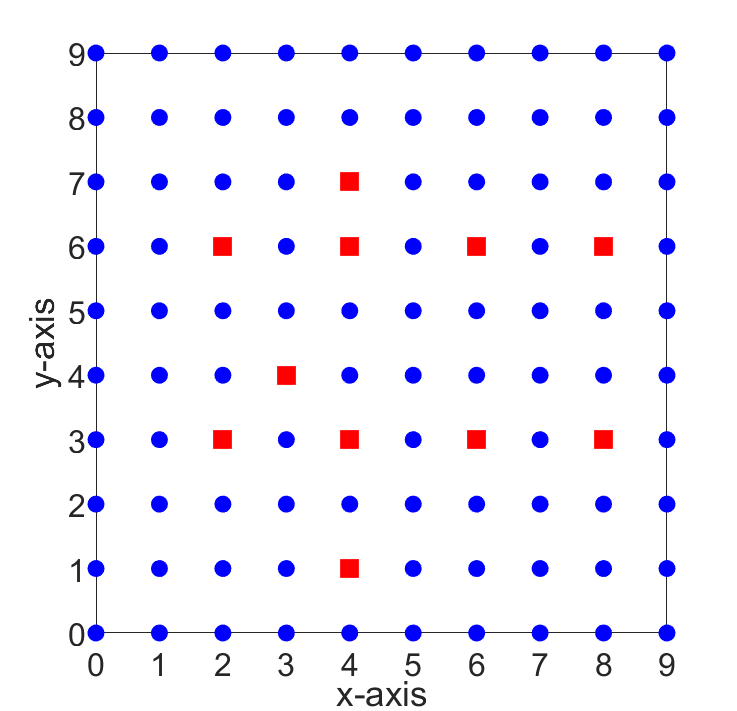}
	\vspace{-5pt}
	\caption{The 100-node sensor network. The red nodes are set as malicious one by one as $f$ increases up to 11.}
	\label{location}
\end{figure}

Here, we examine how the network connectivity affects the performance of the MW-MSR algorithms with different hops. Using different values for the number of malicious agents $f$ and the communication radius $r$, the results of the one-hop algorithm are presented in Fig.~\ref{success}(a). 
For each $f$ and each $r$ (corresponding to one cell in the figure), we compute the success rate of the algorithm to achieve resilient consensus over 50 Monte Carlo runs with randomly chosen initial values (within [0, 100]) of the normal agents for each run. 
The malicious nodes take values based on sine functions w.r.t. time.
Similarly, we also conduct simulations for the two-hop and the three-hop algorithms and the results are given in Figs.~\ref{success}(b) and (c), respectively.

One can see that by increasing the number of hops $l$, the success rate of the algorithm to achieve resilient consensus increases almost for every value of $f$. Such improvement is especially significant when $f\leq6$ and $r\leq3$. This verifies our intuition as well as theoretical findings that graph robustness increases as $l$ increases.
However, the difference between the results for the two-hop and the three-hop algorithms is somewhat minor. One reason is that the maximum robustness under the $f$-total model of a given graph is bounded by the minimum in-degree $2f$. This may indicate that the two-hop communication for this graph already reaches the number of hops for the maximum graph robustness.


\begin{figure}[t]
	\centering
	
	\subfigure[\scriptsize{One-hop algorithm.}]{
		\includegraphics[width=3.2in,height=1.3in]{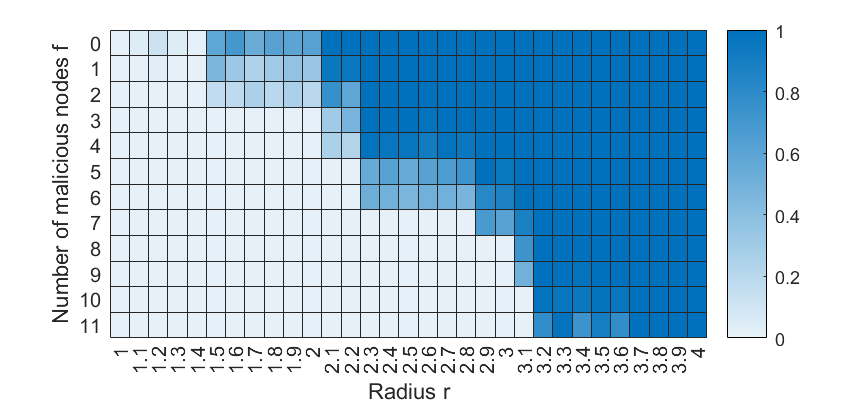}
	}
	
	\vspace{-7pt}
	\subfigure[\scriptsize{Two-hop algorithm.}]{
		\includegraphics[width=3.2in,height=1.3in]{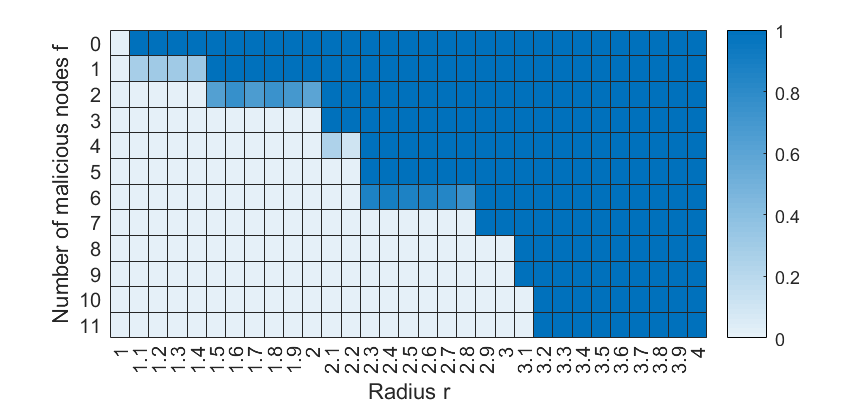}
	}
	\vspace{-7pt}
	\subfigure[\scriptsize{Three-hop algorithm.}]{
		\includegraphics[width=3.2in,height=1.3in]{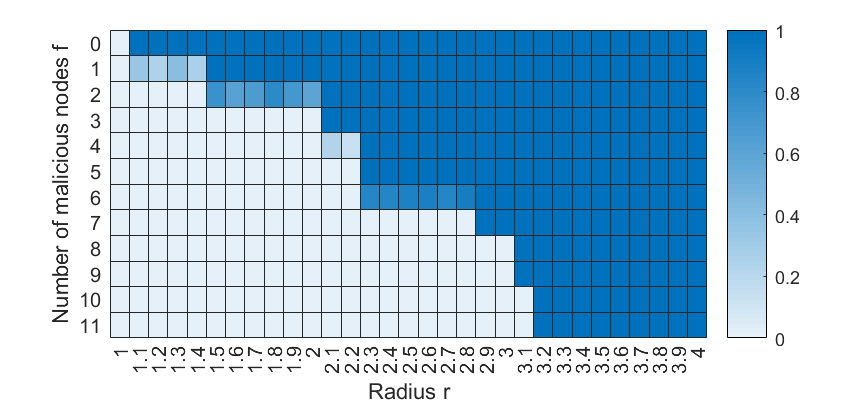}
	}
	\caption{Success rate of the MW-MSR algorithm.}
	\label{success}
	\vspace*{-3.5mm}
\end{figure}

\begin{figure}[t]
	\centering
	
	\subfigure[\scriptsize{Consensus error for $f=0, r=1.2$.}]{
		\includegraphics[width=3.2in,height=1.3in]{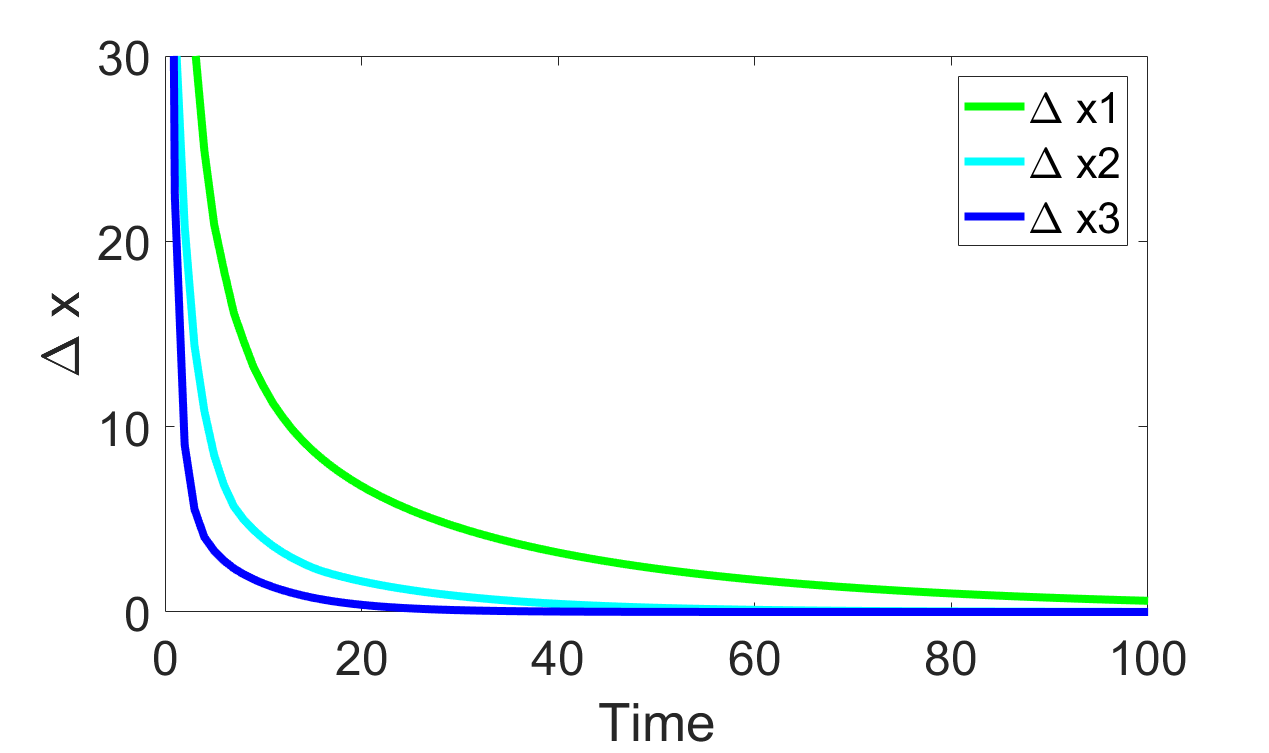}
	}
	\vspace{-7pt}
	\subfigure[\scriptsize{Consensus error for $f=1, r=1.2$.}]{
		\includegraphics[width=3.2in,height=1.3in]{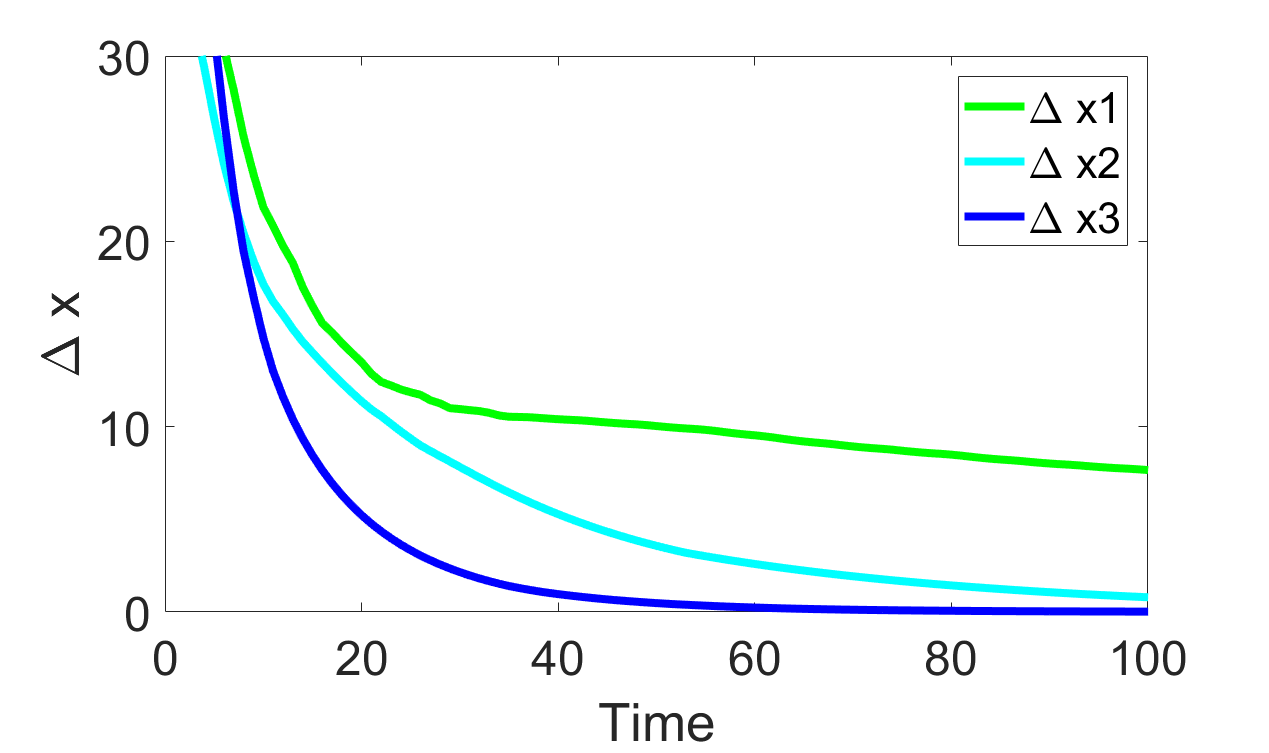}
	}
	\vspace{-7pt}
	\subfigure[\scriptsize{Consensus error for $f=9, r=3.1$.}]{
		\includegraphics[width=3.2in,height=1.3in]{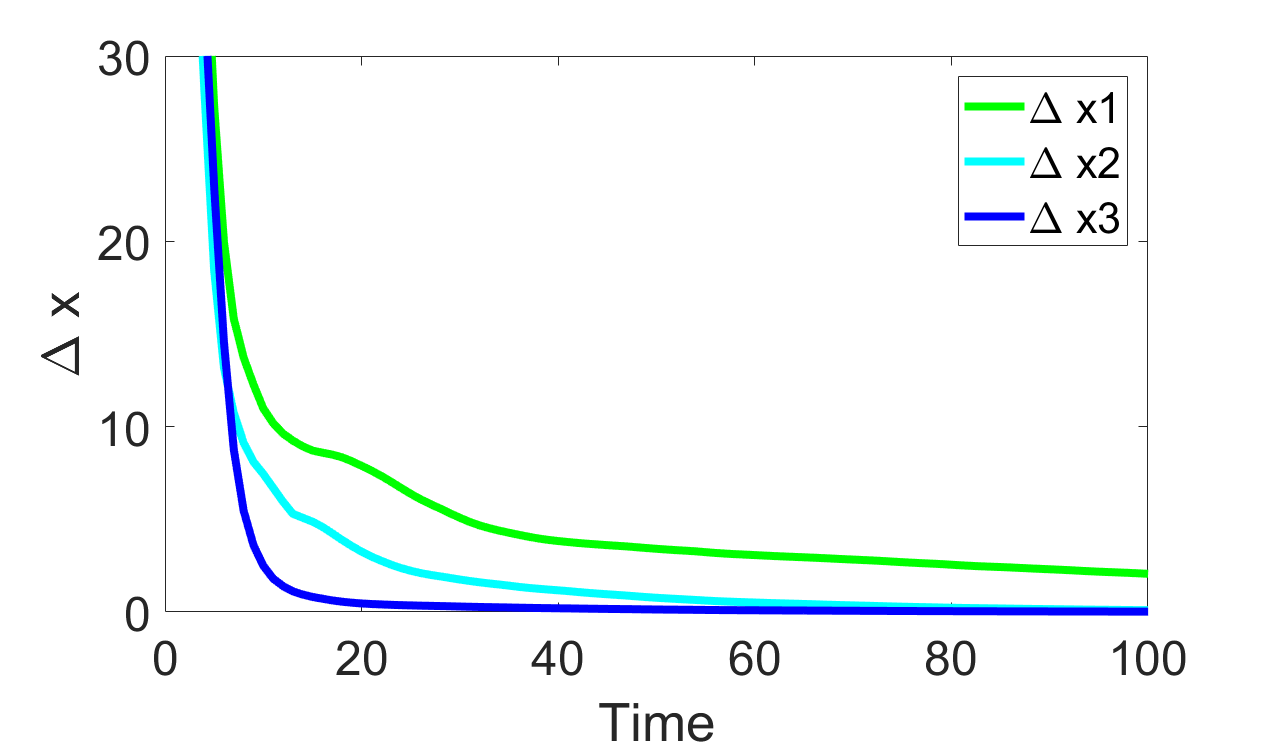}
	}
	
	\caption{Consensus error of the MW-MSR algorithm.}
	\label{error_3algorithms}
	\vspace*{-3.5mm}
\end{figure}

In these simulations, the success rate for reaching consensus is determined by the level of consensus error at time $k=70$. If the consensus error is below the threshold $c=1$, then the run is considered as a success. Hence, if the consensus process is very slow, it may be considered as a failure. 
For example, observe that for the case of $f = 0$ (i.e., with no attacks) the success rate increases with larger $r \leq 1.5$ while the network is connected as long as $r>1$. 

We should remark that the process of consensus forming can be accelerated by increasing the number of hops, as discussed in \cite{jin2006multi} for the fault-free case. This can be clearly seen in all plots in Fig.~\ref{error_3algorithms}, presenting the time responses of the consensus errors for several cases of $f$ and $r$, where $\Delta x1, \Delta x2$ and $ \Delta x3$ stand for the consensus errors for the one-hop, two-hop, and three-hop algorithms, respectively. Based on these examples, we conclude that by introducing multi-hop communication to the MSR algorithms, it does not only improve the robustness of the network but also accelerates the convergence in consensus forming even in adversarial environments.

\section{Conclusion}
In this paper, we have investigated the resilient consensus problem when multi-hop communication is available.
We have proposed generalized versions of MSR algorithms to correctly use the additional values received from multi-hop neighbors.
Moreover, we have fully characterized the network requirement for the algorithms in terms of robustness with $l$ hops. 
By introducing multi-hop communication, the convergence of the resilient consensus process can be accelerated. Furthermore, it provides an effective way to enhance robustness of networks without increasing physical communication links.
In future works, we intend to extend our algorithms to the asynchronous Byzantine consensus problem using multi-hop communication with a fixed number of hops.


\appendices

\section*{Appendix}
\section*{Proof of Theorem \ref{asyntheorem}}

\begin{proof}
	(Necessity) Since the synchronous algorithm is a special case of the asynchronous algorithm, the necessary condition in Theorem \ref{syn} also holds here.

	(Sufficiency) First, we show the safety condition. 
	For $k = 0$, by the assumption on $z[0]$, it holds $z^N[0]\in \mathcal{S}_{\tau}$, and thus $x_i[0] \in \mathcal{S}_{\tau} ,\forall i \in \mathcal{N}$.
	Next, for $k \geq 0$, let $\overline{x}^N_\tau[k]$ and $\underline{x}^N_\tau[k]$ be the largest value and the smallest value, respectively, of the normal agents from time $k, k-1, \dots, k-\tau$. That is, 
	\begin{equation}
	\begin{array}{lll} 
	\overline{x}^N_\tau[k] =\max \left( x^N[k], x^N[k-1],\dots, x^N[k-\tau]\right),\\
	\underline{x}^N_\tau[k] = \min \left( x^N[k], x^N[k-1],\dots, x^N[k-\tau]\right).
	\end{array}
	\end{equation}
	Then we prove that $\overline{x}^N_\tau[k]$ is a nonincreasing function of $k \geq 0$.
	By \eqref{system2}, at time $k \geq 0$, each normal agent updates its value based on a convex combination of the neighbors' values from $k$ to $k-\tau$. 
	Moreover, the values outside of the interval determined by the normal agents' values $[\underline{x}^N_\tau[k], \overline{x}^N_\tau[k]]$ will be ignored by step 2 of the MW-MSR algorithm. This is because in this step, node $i$ will remove the largest sized subsets of large and small values that can be manipulated by at most $f$ nodes within $l$ hops.
	Hence, we obtain $x_i[k+1] \leq \max \left( x^N[k], x^N[k-1],\dots, x^N[k-\tau]\right)$ for any $ i \in \mathcal{N}$. We also have
	\begin{equation*}
	\begin{aligned}
	x_i[k] &\leq \max \left( x^N[k], x^N[k-1],\dots, x^N[k-\tau]\right),\\
	x_i[k-1] &\leq \max \left( x^N[k], x^N[k-1],\dots, x^N[k-\tau]\right),\\
	&\vdots\\
	x_i[k+1-\tau]&\leq \max \left( x^N[k], x^N[k-1],\dots, x^N[k-\tau]\right)
	\end{aligned}
	\end{equation*}
	for any $ i \in \mathcal{N}$. Therefore, $\overline{x}^N_\tau[k]$ is nonincreasing in time as
	\begin{equation*}
	\begin{aligned}
	\overline{x}^N_\tau&[k+1] =\max \left( x^N[k+1], x^N[k],\dots, x^N[k+1-\tau]\right)\\
	&\leq \max \left( x^N[k], x^N[k-1],\dots, x^N[k-\tau]\right)=\overline{x}^N_\tau[k].
	\end{aligned}
	\end{equation*}
	We can similarly prove that $\underline{x}^N_\tau[k]$ is nondecreasing in time. This indicates that for $k \geq 0$, we have $x_i[k]\in \mathcal{S}_{\tau}$, $\forall i \in \mathcal{N}$. Thus, we have shown the safety condition.
	
	Next, we show the convergence. As shown above, $\overline{x}^N_\tau[k]$ and $\underline{x}^N_\tau[k]$ are monotonically decreasing and increasing, respectively, 
	and moreover bounded. Thus, both of their limits exist and are denoted by $\overline{\omega}_\tau$ and $\underline{\omega}_\tau$, respectively.
	We claim that the limits satisfy $\overline{\omega}_\tau=\underline{\omega}_\tau$, i.e., consensus is achieved. We prove by contradiction and assume that $\overline{\omega}_\tau>\underline{\omega}_\tau$.
	
	Recall that $\alpha$ lower bounds the nonzero entries of $\Gamma[k]$. Choose $\epsilon_0 > 0$ small enough that $\overline{\omega}_\tau-\epsilon_0 > \underline{\omega}_\tau+\epsilon_0$. Fix 
	\begin{equation}\label{ep}
	\epsilon<\frac{\epsilon_0\alpha^{(\tau+1)n_N}}{(1-\alpha^{(\tau+1)n_N})}, \medspace 0<\epsilon<\epsilon_0.
	\end{equation}
	Define the sequence $\{\epsilon_\gamma\}$ by
	\begin{equation*}
	\epsilon_{\gamma+1}= \alpha\epsilon_\gamma-(1-\alpha)\epsilon, \medspace \gamma=0,1,\dots,(\tau+1)n_N-1.
	\end{equation*}
	So we have $0 < \epsilon_{\gamma+1} < \epsilon_{\gamma}$ for all $\gamma$. In particular, they are positive because by \eqref{ep}, it holds that
	\begin{equation*}
	\begin{aligned}
	\epsilon_{(\tau+1)n_N}&= \alpha^{(\tau+1)n_N}\epsilon_0- \sum_{m=0}^{(\tau+1)n_N-1}\alpha^m(1-\alpha)\epsilon\\
	&= \alpha^{(\tau+1)n_N}\epsilon_0-(1-\alpha^{(\tau+1)n_N})\epsilon>0.
	\end{aligned}
	\end{equation*}
	
	Take $k_\epsilon \in \mathbb{Z}_+$ such that $\overline{x}^N_\tau[k]<\overline{\omega}_\tau+\epsilon$ and $\underline{x}^N_\tau[k]>\underline{\omega}_\tau-\epsilon$ for $k\geq k_\epsilon$. Such $k_\epsilon$ exists due to the convergence of $\overline{x}^N_\tau[k]$ and $\underline{x}^N_\tau[k]$. Then we can define the two disjoint sets as 
	\begin{equation*}
	\mathcal{Z}_{1\tau}(k_\epsilon+\gamma,\epsilon_\gamma)=\{j\in \mathcal{N}: x_j[k_\epsilon+\gamma]>\overline{\omega}_\tau-\epsilon_\gamma\},
	\end{equation*}
	\begin{equation*}
	\mathcal{Z}_{2\tau}(k_\epsilon+\gamma,\epsilon_\gamma)=\{j\in \mathcal{N}: x_j[k_\epsilon+\gamma]<\underline{\omega}_\tau+\epsilon_\gamma\}.
	\end{equation*}
	
	Next, we show that one of the two sets becomes empty in a finite number of steps, which contradicts the assumption on $\overline{\omega}_\tau$ and $\underline{\omega}_\tau$ being the limits. Consider the set $\mathcal{Z}_{1\tau}(k_\epsilon,\epsilon_0)$. Due to the definition of $\overline{x}^N_\tau[k]$ and its limit $\overline{\omega}_\tau$, one or more normal nodes are contained in the union of the sets $\mathcal{Z}_{1\tau}(k_\epsilon+\gamma,\epsilon_\gamma)$ for $0 \leq\gamma\leq\tau + 1$. We claim that $\mathcal{Z}_{1\tau}(k_\epsilon,\epsilon_0)$ is in fact nonempty. To prove this, it is sufficient to show that if a normal node $j$ is not in $\mathcal{Z}_{1\tau}(k_\epsilon+\gamma,\epsilon_\gamma)$, then it is not in $\mathcal{Z}_{1\tau}(k_\epsilon+\gamma+1,\epsilon_{\gamma+1})$ for $\gamma=0,\dots,\tau$.
	
	Suppose that node $j$ satisfies $x_j[k_\epsilon+\gamma]\leq \overline{\omega}_\tau-\epsilon_\gamma$. Every normal node updates its value to a convex combination of the multi-hop neighbors' values at the current or previous times. Moreover, the values greater than $\overline{x}^N_\tau[k_\epsilon+\gamma]$ are ignored in step 2 of the MW-MSR algorithm. Hence, the value of node $j$ at the next time step is upper bounded as
	\begin{equation}\label{converge}
	\begin{aligned}
	x_j&[k_\epsilon+\gamma+1] \leq (1-\alpha)\overline{x}^N_\tau[k_\epsilon+\gamma]+\alpha(\overline{\omega}_\tau-\epsilon_\gamma)\\
	&\leq (1-\alpha)(\overline{\omega}_\tau+\epsilon)+\alpha(\overline{\omega}_\tau-\epsilon_\gamma)\\
	&\leq \overline{\omega}_\tau-\alpha\epsilon_\gamma+(1-\alpha)\epsilon
	=\overline{\omega}_\tau-\epsilon_{\gamma+1}.
	\end{aligned}
	\end{equation}
	It thus follows that node $j$ is not in $\mathcal{Z}_{1\tau}(k_\epsilon+\gamma+1,\epsilon_{\gamma+1})$. This means that the cardinality of the set $\mathcal{Z}_{1\tau}(k_\epsilon+\gamma,\epsilon_\gamma)$ is nonincreasing for $\gamma=0,\dots,\tau+1$. The same holds for $\mathcal{Z}_{2\tau}(k_\epsilon+\gamma,\epsilon_\gamma)$, and hence $\mathcal{Z}_{2\tau}(k_\epsilon,\epsilon_0)$ is nonempty too.

	We next show that one of these two sets in fact becomes empty in finite time. Since the graph is $(2f + 1)$-robust with $l$ hops w.r.t. any set $\mathcal{F}$ satisfying the $f$-total model, the graph is also $(2f + 1)$-robust with $l$ hops w.r.t. set $\mathcal{A}$ (i.e., the set of adversarial nodes).
	Therefore, between the two nonempty disjoint sets $\mathcal{Z}_{1\tau}(k_\epsilon,\epsilon_0)$ and $\mathcal{Z}_{2\tau}(k_\epsilon,\epsilon_0)$, one of them has a normal agent with at least $2f + 1$ independent paths originating from the nodes outside and these paths do not have any internal node in the set $\mathcal{A}$.

	Suppose that normal
	node $i\in \mathcal{Z}_{1\tau}(k_\epsilon,\epsilon_0)$ has this property. Since there are at most $f$ malicious nodes and node $i$ can only remove the values of which the cardinality of the minimum message cover is $f$. Moreover, node $i$ is supposed to update once in at most $\tau$ time steps.
	Therefore, when node $i$ makes an update at time $k_\epsilon+\tau$, it will use at least one delayed value from the normal nodes outside the set  $\mathcal{Z}_{1\tau}(k_\epsilon,\epsilon_0)$, upper bounded by $\overline{\omega}_\tau-\epsilon_\tau$. It thus follows that, at time $k_\epsilon+\tau$, when node $i$ makes an update, its value can be bounded as
	\begin{equation*}
	x_i[k_\epsilon+\tau+1]\leq (1-\alpha)\overline{x}^N_\tau[k_\epsilon+\tau]+\alpha(\overline{\omega}_\tau-\epsilon_\tau).
	\end{equation*}
	By \eqref{converge}, we have $x_i[k_\epsilon+\tau+1] \leq \overline{\omega}_\tau-\epsilon_{\tau+1}$. We can conclude that if node $i$ in $\mathcal{Z}_{1\tau}(k_\epsilon,\epsilon_0)$ has $2f +1$ independent paths originating from the nodes outside the set, then it goes outside of $\mathcal{Z}_{1\tau}(k_\epsilon+\tau+1,\epsilon_{\tau+1})$ after $\tau + 1$ steps. Consequently, 
	$\left|\mathcal{Z}_{1\tau}(k_\epsilon+\tau+1,\epsilon_{\tau+1}) \right| < \left| \mathcal{Z}_{1\tau}(k_\epsilon,\epsilon_0)\right| $. Likewise, it follows that if $\mathcal{Z}_{2\tau}(k_\epsilon,\epsilon_0)$ has a node having at least $2f+1$ independent paths originating from the nodes outside, then $\left|\mathcal{Z}_{2\tau}(k_\epsilon+\tau+1,\epsilon_{\tau+1}) \right| < \left| \mathcal{Z}_{2\tau}(k_\epsilon,\epsilon_0)\right| $.
	
	Since there are only $n_N$ normal nodes, we can repeat the steps above until one of the sets $\mathcal{Z}_{1\tau}(k_\epsilon+\tau+1,\epsilon_{\tau+1})$ and $\mathcal{Z}_{2\tau}(k_\epsilon+\tau+1,\epsilon_{\tau+1})$ becomes empty, and it takes no more than $(\tau + 1)n_N$ steps. Once the set becomes empty, it remains so indefinitely. This contradicts the assumption that $\overline{\omega}_\tau$ and $\underline{\omega}_\tau$ are the limits. Therefore, we obtain $\overline{\omega}_\tau=\underline{\omega}_\tau$.
\end{proof}

\begin{IEEEbiography}[{\includegraphics[width=1in,height=1.25in,clip,keepaspectratio]{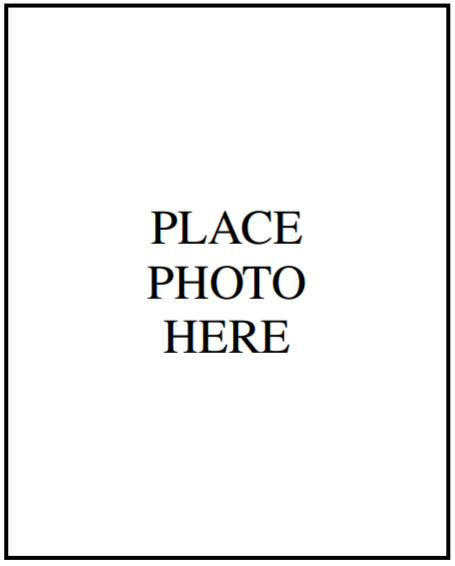}}]{Liwei Yuan} received the B.E. degree in Electrical Engineering and Automation from Tsinghua University, China, in 2017, and the M.E. degree in Computer Science from Tokyo Institute of Technology, Japan, in 2019. He is currently working toward the Ph.D. degree in the Department of Computer Science at Tokyo Institute of Technology, Yokohama, Japan. His current research focuses on security in multi-agent systems and distributed algorithms.
\end{IEEEbiography}

\begin{IEEEbiography}[{\includegraphics[width=1in,height=1.25in,clip,keepaspectratio]{blank}}]{Hideaki Ishii} (M'02-SM'12-F'21) received the
	M.Eng.\ degree from Kyoto University, Kyoto, Japan, in 1998, and the
	Ph.D. degree from the University of Toronto, Toronto,
	ON, Canada, in 2002. He was a Postdoctoral Research
	Associate at the University of Illinois at Urbana-Champaign,
	Urbana, IL, USA, from 2001 to
	2004, and a Research Associate at The University of Tokyo, Tokyo, Japan, 
	from 2004 to 2007.
	Currently, he is a Professor at the Department of Computer Science,
	Tokyo Institute of Technology, Yokohama, Japan. 
	His research interests
	are in networked control systems, multiagent systems, distributed algorithms,
	and cyber-security of control systems.
	
	Dr.~Ishii has served as an Associate Editor for the IEEE Control
	Systems Letters and the Mathematics of Control, Signals, and Systems
	and previously for Automatica, the IEEE Transactions on Automatic
	Control, and the IEEE Transactions on Control of Network Systems.
	He is the Chair of the IFAC Coordinating Committee on Systems and Signals since 2017.
	He is the IPC Chair for the IFAC World Congress 2023 to be held in Yokohama, Japan.
	He received the IEEE Control Systems Magazine Outstanding Paper
	Award in 2015. Dr.~Ishii is an IEEE Fellow.
\end{IEEEbiography}

\end{document}